\documentclass[journal]{IEEEtran}
\usepackage{mathrsfs}
\usepackage{cite}
\usepackage{color}
\usepackage{psfrag}
\usepackage{subfigure}
\usepackage{amssymb}
\usepackage{epsfig}
\usepackage{pifont}
\usepackage{amsmath}
\usepackage{array}
\usepackage{pifont}
\usepackage{amsfonts}
\usepackage{fancyhdr}
\usepackage{amscd}
\usepackage{bm}
\usepackage{amssymb}
\usepackage{graphics,eepic,epic}
\usepackage{latexsym}
\usepackage{euscript}
\usepackage{graphics,eepic,epic,psfrag}
\usepackage{amsthm}

\newtheorem{proposition}{Proposition}

\usepackage{algorithm}
\usepackage{algpseudocode}
\newcommand{\harpoonl}{\overset{\leftharpoonup}}
\newcommand{\harpoonr}{\overset{\rightharpoonup}}

\begin{document}
\title{Sparse Activity Detection for Massive Connectivity}
\author{Zhilin~Chen,~\IEEEmembership{Student Member,~IEEE,}
        Foad~Sohrabi,~\IEEEmembership{Student Member,~IEEE,}
        and~Wei~Yu,~\IEEEmembership{Fellow,~IEEE}
\thanks{Manuscript accepted and to appear in IEEE Transactions on Signal Processing. This work has been presented in part at IEEE International Conference on Acoustics, Speech, and Signal Processing (ICASSP), March 2017. This work is supported by Natural Sciences and Engineering Research Council (NSERC) of Canada through a Discovery Grant and through a Steacie Memorial Fellowship.\par
The authors are with The Edward S. Rogers Sr.  Department of Electrical and
Computer Engineering, University of Toronto, Toronto, ON M5S 3G4, Canada
(e-mails:\{zchen, fsohrabi, weiyu\}@comm.utoronto.ca).}}
\maketitle
\begin{abstract}
This paper considers the massive connectivity application in which a large number
of potential devices communicate with a base-station (BS) in a sporadic fashion.
The detection of device activity pattern together with the estimation of the
channel are central problems in such a scenario.  Due to the large number of
potential devices in the network, the devices need to be assigned
non-orthogonal signature sequences. The main objective of this paper is to show
that by using random signature sequences and by exploiting sparsity in the user
activity pattern, the joint user detection and channel estimation problem can
be formulated as a compressed sensing single measurement vector (SMV) problem
or multiple measurement vector (MMV) problem, depending on whether the BS has a
single antenna or multiple antennas, and be efficiently solved using an approximate
message passing (AMP) algorithm. This paper proposes an AMP algorithm design
that exploits the statistics of the wireless channel and provides an
analytical characterization of the probabilities of false alarm and missed
detection by using the state evolution. We consider two cases depending on
whether the large-scale component of the channel fading is known at the BS and
design the minimum mean squared error (MMSE) denoiser for AMP according to the
channel statistics. Simulation results demonstrate the substantial advantage of
exploiting the statistical channel information in AMP design; however, knowing
the large-scale fading component does not offer tangible benefits. For the
multiple-antenna case, we employ two different AMP algorithms, namely the AMP
with vector denoiser and the parallel AMP-MMV, and quantify the benefit of
deploying multiple antennas at the BS.
\end{abstract}
\IEEEpeerreviewmaketitle

\section{Introduction}
One of the key requirements for the next-generation wireless cellular
networks is to provide massive connectivity for machine-type
communications (MTC), envisioned to support diverse applications such
as environment sensing, event detection, surveillance and control
\cite{andrews2014will, Wei2016}.
Machine-centric communications have two distinctive features as
compared to conventional human-centric communications: (i) the overall system needs
to support massive connectivity---the number of devices connected to
each cellular base-station (BS) may be in the order of $10^4$ to
$10^6$; and (ii) the traffic pattern
of each device may be sporadic---at any given time only a small
fraction of potential devices are active. For such a network,
accurate user activity detection and channel estimation are crucial for
establishing successful communications between the devices and the BS.

To identify active users and to estimate their channels, each
user must be assigned a unique signature sequence. However, due to the
large number of potential devices but the limited coherence time and
frequency dimensions in the wireless fading channel, the signature
sequences for all users cannot be mutually orthogonal.
Non-orthogonal signature sequences superimposed in the pilot stage
causes significant multi-user interference, e.g., when a simple matched
filtering or correlation operation is applied at the BS for user
activity detection and channel estimation. A key observation of this
paper is that the sporadic nature of the traffic leads to
\emph{sparse} user transmission patterns. By exploiting sparsity and by
formulating the detection and estimation problem with independent identically distributed (i.i.d.) random non-orthogonal
pilots as a compressed sensing problem, this multi-user interference
problem can be overcome, and highly reliable activity detection and
accurate channel estimation can be made possible.
In the compressed sensing terminology, when the BS is equipped with a
single antenna, activity detection and channel estimation can be
formulated as a single measurement vector (SMV)
problem; when the BS has multiple antennas, the problem can be
formulated as a multiple measurement vector (MMV) problem.

This paper proposes the use of compressed sensing techniques for
the joint user activity detection and channel estimation problem.
Due to the large-scale nature of massive device communications, this paper
adopts the computationally efficient approximate message passing (AMP)
algorithm \cite{Donoho2009} as the main technique. AMP is an iterative
thresholding method with a key feature that allows analytic performance
characterization via the so-called state evolution. The main
contributions of this paper are: (i) a novel AMP algorithm design for
user activity detection that exploits the statistical information of the
wireless channel; and (ii) a characterization of the probabilities of
false alarm and missed detection for both SMV and MMV scenarios.

\subsection{Related Work}


The user activity detection problem for massive connectivity has been studied
from information theoretical perspectives in \cite{ChenGuo2017,Wei2016}.
From an algorithmic point of view, the problem is closely related to sparse
recovery in compressed sensing and has been studied in a variety of wireless
communication settings. For example, assuming no prior knowledge of the
channel state information (CSI), joint user activity detection and channel
estimation is considered in \cite{Dekorsy2013,Lau2015,Wunder2015GC,Wunder2015}.
Specifically, \cite{Dekorsy2013} proposes an efficient greedy algorithm based
on orthogonal matching pursuit for sporadic multi-user communication. By
exploiting the statistics of channel path-loss and the joint sparsity
structures, \cite{Lau2015} proposes a modified Bayesian compressed sensing
algorithm in a cloud radio-access network. In the context of orthogonal
frequency division multiplexing (OFDM) systems, \cite{Wunder2015GC} introduces
a one-shot random access protocol and employs the basis pursuit denoising
detection method with a detection error bound based on the restricted isometry
property. The performance of such schemes in a practical setting is illustrated
in \cite{Wunder2015GC,Wunder2015}. When perfect CSI is assumed, joint user
activity and data detection for code division multiple access systems (CDMA)
is investigated in \cite{ZhuGiannakis2011,SchepkerDekorsy2012}, where
\cite{ZhuGiannakis2011} designs the sparsity-exploiting maximum a posteriori
detector by accounting for both sparsity and finite-alphabet constraints of the
signal, and \cite{SchepkerDekorsy2012} proposes a greedy block-wise orthogonal
least square algorithm by exploiting the block sparsity among several symbol
durations. Differing from most of the above works that consider cellular
systems, \cite{LuoGuo2008,Guo2013} study the user activity detection in
wireless ad hoc networks, where each node in the system identifies its neighbor
nodes simultaneously. The authors of \cite{LuoGuo2008} propose a scalable
compressed neighbor discovery scheme that employs random binary signatures and
group testing based detection algorithms. In \cite{Guo2013}, the authors
propose a more dedicated scheme that uses signatures based on
Reed-Muller code and a chirp decoding algorithm to achieve a better
performance.

In contrast to the aforementioned works, this paper adopts the more
computationally efficient AMP algorithm for user activity detection and channel
estimation, which is more suitable for
large-scale networks with a large number of devices. The AMP algorithm is first
proposed in \cite{Donoho2009} as a low-complexity iterative algorithm for
conventional compressed sensing with real-valued signals and real-valued
measurements. A framework of state evolution that tracks the performance of AMP
at each iteration is introduced in \cite{Donoho2009}. The AMP algorithm is then
extended along different directions. For example,
\cite{Bayati2011} generalizes the AMP algorithm to a broad family of iterative
thresholding algorithms, and provides a rigorous proof of the framework of the
state evolution. To deal with complex-valued signals and measurements,
\cite{Maleki2013} proposes a complex AMP algorithm (CAMP). By exploiting the
input and output distributions, a generalized AMP (GAMP) algorithm is designed in \cite{Rangan2011}. Similarly, a Bayesian approach is used to design the AMP algorithm in \cite{Donoho2010,Schniter2010} by accounting for the input distribution.
For the compressed sensing problem with multiple signals sharing joint
sparsity, i.e., the MMV problem,
\cite{Kim2011} designs an AMP algorithm via a vector form of message passing;
and \cite{Ziniel2013} designs the AMP-MMV algorithm by directly using message
passing over a multi-frame factor graph.

Although the use of the AMP algorithm for user activity detection has
been previously proposed in \cite{Goertz2015}, the statistical
information of the channel is not exploited in the prior work; also performance
analysis is not yet available. This paper makes progress by showing that
exploiting channel statistics can significantly enhance the Bayesian AMP
algorithm.  Moreover, analytical performance characterization can be obtained
by using state evolution. Finally, the AMP algorithm can be extended to the
multiple-antenna case.

\subsection{Main Contributions}

This paper considers the user activity detection and channel estimation problem
in the uplink of a single-cell network with a large number of potential users,
but at any given time slot only a small fraction of them are active. To exploit
the sparsity in user activity pattern, this paper formulates the problem as a
compressed sensing problem and proposes the use of random signature sequences
and the computationally efficient AMP algorithm for device activity detection.
This paper provides the design and analysis of AMP for both cases in which
the BS is equipped with a single antenna and with multiple antennas.

This paper considers two different scenarios: (i) when the large-scale fading
coefficients of all user are known and the detector is designed based on the
statistics of fast fading component only; and (ii) when the large-scale fading
coefficients are not known and the detector is designed based on the statistics
of both fast fading and large-scale fading components as a function of the
distribution of device locations in the cell.  The proposed AMP-based detector
exploits the statistics of the wireless channel by specifically designing the
minimum mean squared error (MMSE) denoiser.  This paper provides analytical
characterization of the probabilities of false alarm and missed detection via
the state evolution for both scenarios.

For the case where the BS is equipped with a single antenna, numerical results
indicate that: (i) the analytic performance characterization via state
evolution is very close to the simulation; (ii) exploiting the statistical
information of the channel and user activity can significantly improve the
detector performance; and (iii) knowing the large-scale fading coefficient
actually does not bring substantial performance improvement as compared to the
case that only the statistical information about the large-scale fading is
available.

For the case where the BS is equipped with multiple antennas, this paper
considers both the AMP with vector denoiser \cite{Kim2011} and the parallel
AMP-MMV \cite{Schniter2010}. For the AMP with vector denoiser, this paper
exploits wireless channel statistics in denoiser design and further
analytically characterizes the probabilities of false alarm and missed
detection based on the state evolution. For the parallel AMP-MMV algorithm,
which is more suitable for distributed computation, performance
characterization is more difficult to obtain. Simulation results show that:
(i) having multiple antennas at the BS can significantly improve the detector
performance; (ii) the predicted performance of AMP with vector denoiser is very
close to its simulated performance; and (iii) AMP with vector denoiser and
parallel AMP-MMV achieve approximately the same performance.

\subsection{Paper Organization and Notations}

The remainder of this paper is organized as follows.
Section~\ref{sec.system} introduces the system model.
Section~\ref{sec.amp} introduces the AMP algorithms
for both SMV and MMV problems. Section~\ref{sec.single}
considers user activity detection and channel estimation when the BS has a single-antenna, while Section~\ref{sec.multiple} considers the multiple-antenna case. Simulation results are provided in
Section~\ref{sec.numerical}. Conclusions are drawn in Section~\ref{sec.con}.

Throughout this paper, upper-case and lower-case letters denote random variables and their realizations, respectively. Boldface lower-case letters denote vectors. Boldface upper-case letters denote matrices or random vectors, where context should make the distinction clear. Superscripts $(\cdot)^{T}$, $(\cdot)^{*}$ and $(\cdot)^{-1}$ denote transpose, conjugate transpose, and inverse operators, respectively. Further, $\mathbf{I}$ denotes identity matrix with appropriate dimensions, $\mathbb{E}[\cdot]$ denotes expectation operation, $\triangleq$ denotes definition, $|\cdot|$ denotes either the magnitude of a complex variable or the determinant of a matrix, depending on the context, and $\|\cdot\|_2$ denotes the $\ell_2$ norm.

\section{System Model}
\label{sec.system}
Consider the uplink of a wireless cellular system with one BS located at the
center and $N$ single-antenna devices located uniformly in a circular area with
radius $R$, but in each coherence block only a subset of users are active.  Let
$a_{n}\in\{1,0\}$ indicate whether or not user $n$ is active.  For the purpose
of channel probing and user identification, user $n$ is assigned a unique
signature sequence $\mathbf{s}_n=[s_{1n}, s_{2n},\cdots,s_{Ln}]^{T}\in
\mathbb{C}^{L\times 1}$, where $L$ is the length of the sequence.
This paper assumes that the signature sequence $\mathbf{s}_n$ is generated
according to i.i.d. complex Gaussian distribution with zero mean and variance
$1/L$ such that each sequence is normalized to have unit power, and the
normalization factor $1/L$ is incorporated into the transmit power.

We consider a block-fading channel model where the channel is static in each
block. In this paper, we consider two cases where the BS is equipped with
either a single antenna or multiple antennas.
When the BS has only one antenna, the received signal at the BS can be
modeled as
\begin{equation}\label{eq.signal}
\mathbf{y} = \sum_{n=1}^{N}a_n\mathbf{s}_{n}h_n + \mathbf{w}\triangleq \mathbf{Sx}+\mathbf{w},
\end{equation}
where $h_n\in\mathbb{C}$ is the channel coefficient between user $n$
and the BS, $\mathbf{w}\in\mathbb{C}^{L\times 1}$ is the effective
complex Gaussian noise whose variance $\sigma^{2}_{w}$ depends on the
background noise power normalized by the user transmit power. Here,
$\mathbf{x}\triangleq[x_1,x_2,\cdots,x_N]^{T}\in\mathbb{C}^{N\times 1}$
where $x_n\triangleq h_na_n$, and
$\mathbf{S}\triangleq[\mathbf{s}_1,\mathbf{s}_2,\cdots,\mathbf{s}_N]\in\mathbb{C}^{L\times N}$.

We aim to recover the
non-zero entries of $\mathbf{x}$ based on the received signals
$\mathbf{y}$. We are interested in the regime where the number of
potential users is much larger than the pilot sequence length, i.e.,
$N \gg L$, so that the user pilot sequences cannot be mutually
orthogonal; but due to the sporadic traffic, only a small number of
devices transmit in each block, resulting in a sparse $\mathbf{x}$.
The recovering of $\mathbf{x}$ for the single antenna case is in the form of the SMV problem in
compressed sensing.

This paper also considers the case where the BS is equipped with $M$
antennas. In this case, the received signal $\mathbf{Y}\in\mathbb{C}^{L\times M}$ at
the BS can be expressed in matrix form as
\begin{equation}\label{eq.signal.multiple}
\mathbf{Y} = \sum_{n=1}^{N}a_n\mathbf{s}_{n}\mathbf{h}_{n} + \mathbf{W}\triangleq \mathbf{SX}+\mathbf{W},
\end{equation}
where $\mathbf{h}_{n}\in\mathbb{C}^{1\times M}$ is the channel vector between user $n$
and the BS, $\mathbf{W}\in\mathbb{C}^{L\times M}$ is the effective
complex Gaussian noise, and
$\mathbf{X}\triangleq[\mathbf{r}_{1}^{T},\cdots,\mathbf{r}_{N}^{T}]^{T}\in\mathbb{C}^{N\times M}$ where $\mathbf{r}_{n}\triangleq a_n\mathbf{h}_{n}\in\mathbb{C}^{1\times M}$ is the $n$th row vector of $\mathbf{X}$. We also use $\mathbf{c}_m\in\mathbb{C}^{N\times 1}$ to represent the $m$th column vector of $\mathbf{X}$, i.e., $\mathbf{X}=[\mathbf{c}_{1},\cdots,\mathbf{c}_{M}]$.
Note that $a_n$ indicates whether the entire row vector $\mathbf{r}_{n}$ is zero or not.
In other words, columns of $\mathbf{X}$ (i.e., $\mathbf{c}_m$) share the same sparsity pattern.

We are interested in detecting the user activity $a_n$ as well as in
estimating the channel gains of the active users, which correspond to the non-zero
rows of the matrix $\mathbf{X}$, based on the observation $\mathbf{Y}$
in the regime where $N \gg L$. The problem of recovering $\mathbf{X}$
from $\mathbf{Y}$ is in the form of the MMV problem in compressed sensing.

A key observation of this paper is that the design of recovery algorithm can be significantly enhanced by taking advantage
of the knowledge about the statistical information of $\mathbf{x}$ or $\mathbf{X}$. Toward this
end, we provide a model for the distribution of the entries of $\mathbf{x}$, and the distribution of the rows of $\mathbf{X}$. Since $\mathbf{x}$ is a special case of $\mathbf{X}$ when $M=1$, we focus on the model for $\mathbf{X}$.

We assume that each user accesses the channel with a small probability $\lambda$
in an i.i.d. fashion, i.e.,  $\mathrm{Pr}(a_n=1)=\lambda, \forall n$, and there is no correlation between different users' channels,
so that the row vectors of $\mathbf{X}$ follow a mixture distribution
\begin{align}\label{eq.distr.multiple}
p_{\mathbf{R}|G}(\mathbf{r}_n|g_n)=(1-\lambda)\delta_{\mathbf{0}}+\lambda p_{\mathbf{H}|G}(\mathbf{r}_n|g_n),
\end{align}
where $\delta_{\mathbf{0}}$ denotes the point mass measure at $\mathbf{0}$, $p_{\mathbf{H}|G}$
denotes the probability density function (pdf) of the
channel vector $\mathbf{H}$ given prior information $G$, which has a pdf
$p_G$, and $g_n$
denotes the prior information for user $n$. Note that we use $\mathbf{H}$ to denote the random channel vector and $\mathbf{h}_n$ to denote its realization. Based on \eqref{eq.distr.multiple}, the pdf of the entries of $\mathbf{x}$ is
\begin{equation}\label{eq.distr}
p_{X|G}(x_n|g_n)=(1-\lambda)\delta_{0}+\lambda p_{H|G}(x_n|g_n).
\end{equation}

To model the distribution of $\mathbf{H}$, we assume that all users
are randomly and uniformly located in a circular coverage area of
radius $R$ with the BS at the center, and the channels between the
users and the BS follow an independent distribution that depends on
the distance. More specifically,
$\mathbf{H}$ includes path-loss, shadowing, and Rayleigh fading.
The path-loss between a user and the BS is modeled (in dB) as
$\alpha+\beta \log_{10}(d)$, where $d$ is the distance
measured in meter, $\alpha$ is the fading coefficient at $d=1$, and
$\beta$ is the path-loss exponent. The shadowing (in dB) follows a
Gaussian distribution with zero mean and variance
$\sigma_{\mathrm{SF}}^2$. The Rayleigh fading is assumed to be i.i.d. complex Gaussian with zero mean and unit variance across all antennas.

The large-scale fading, which includes path-loss and shadowing, is denoted as
$G$, whose pdf $p_G$ can be modeled by the distribution of BS-user distance and
shadowing parameter $\sigma_{\mathrm{SF}}^2$.  This paper considers both the
case where only the statistics of the the large-scale fading, i.e., $p_G$, is
known as well as the case where the exact large-scale fading coefficient $g_n$
is known at the BS.  The latter case is motivated by the scenario in which the
devices are stationary, so that the
path-loss and shadowing can be estimated and stored at the BS as prior information. When $g_n$ is known, $p_{\mathbf{H}|G}$ captures the distribution of the Rayleigh fading component. When only $p(G)$ is known, we drop $G$ and $g_n$ from $p_{\mathbf{H}|G}(\mathbf{h}_n|g_n)$, and write it as $p_{\mathbf{H}}(\mathbf{h}_n)$, which captures the distribution of both large-scale fading and Rayleigh fading.

\section{AMP Algorithm}
\label{sec.amp}

AMP is an iterative algorithm that recovers sparse signal for compressed
sensing. We introduce the AMP framework for both the SMV and the MMV problems
in this section.

\subsection{AMP for SMV problem}
AMP is first proposed for the SMV problem in \cite{Donoho2009}.
Starting with $\mathbf{x}^{0}=\mathbf{0}$ and
$\mathbf{z}^{0}=\mathbf{y}$, AMP proceeds at each
iteration as
\begin{align}
\mathbf{x}^{t+1}&=\eta(\mathbf{S}^{*}\mathbf{z}^{t}+\mathbf{x}^{t},\mathbf{g},t), \label{eq.amp1}\\
\mathbf{z}^{t+1} &= \mathbf{y}-\mathbf{S}\mathbf{x}^{t+1}+\frac{N}{L}\mathbf{z}^{t}\langle\eta^{\prime}(\mathbf{S}^{*}\mathbf{z}^{t}+\mathbf{x}^{t},\mathbf{g},t)\rangle, \label{eq.amp2}
\end{align}
where $\mathbf{g}\triangleq [g_1,\cdots,g_N]^{T}$, and $t=0,1,\cdots$ is the index of iteration, $\mathbf{x}^{t}$ is
the estimate of $\mathbf{x}$ at iteration $t$, $\mathbf{z}^{t}$ is the
residual, $\eta(\cdot,\mathbf{g},t)\triangleq [\eta_t(\cdot,g_1),\cdots,\eta_{t}(\cdot,g_N)]^{T}$ where $\eta_t(\cdot,g_n): \mathbb{C}\rightarrow \mathbb{C}$ is an appropriately designed non-linear function known as \emph{denoiser} that operates on the $n$th entry of the input vector, $\eta^{\prime}(\cdot)\triangleq[\eta_t^{\prime}(\cdot,g_1),\cdots,\eta_{t}^{\prime}(\cdot,g_N)]^{T}$ where $\eta_t^{\prime}(\cdot,g_n)$ is the
first order derivative of $\eta_t(\cdot,g_n)$ with respect to the first argument, and $\langle\cdot\rangle$ is
averaging operation over all entries of a vector. Note that the third
term in the right hand side of \eqref{eq.amp2} is the correction term
known as the ``Onsager term'' from statistical physics.

In the AMP algorithm, the matched filtered output
$\mathbf{\tilde{x}}^{t}\triangleq\mathbf{S}^{*}\mathbf{z}^{t}+\mathbf{x}^{t}$ can be modeled as signal $\mathbf{x}$ plus
noise (including multiuser interference), i.e.,  $\mathbf{\tilde{x}}^{t}=\mathbf{x}+\mathbf{v}^{t}$, where $\mathbf{v}^{t}$ is Gaussian due to the correction term.
The denoiser is typically designed to reduce the estimation error at each iteration.
In the compressed sensing literature, the prior distribution of
$\mathbf{x}$ is usually assumed to be unknown. In this case, a
minimax framework over the worst case $\mathbf{x}$ leads to a soft
thresholding denoiser
\cite{Donoho2013Minimax}. When the prior
distribution of $\mathbf{x}$ is known, the Bayesian framework then can
be used to account for the prior information on $\mathbf{x}$\cite{Donoho2010}. In
this paper, we adopt the Bayesian approach and design the MMSE
denoiser for the massive connectivity setup as shown in the next section.

The AMP algorithm can be analyzed in the asymptotic regime where
$L,N\rightarrow\infty$ with fixed $N/L$ via the state evolution, which
predicts the per-coordinate performance of the AMP algorithm at each
iteration as follows
\begin{align}\label{eq.stateevo}
\tau_{t+1}^2 = \sigma_w^2+\frac{N}{L}\mathbb{E}\big[\left|\eta_t(X+\tau_tV,G)-X\right|^2\big],
\end{align}
where $\tau_{t}$ is referred to as the \emph{state}, $X$, $V$, and $G$ are
random variables with $X$ following $p_{X|G}$, $V$ following the
complex Gaussian distribution with zero mean and unite variance, and $G$ following $p_G$, and the expectation is taken
over all $X$, $V$, and $G$. We denote $\tilde{X}^{t}\triangleq X+\tau_{t}V$. The random variables $X$, $V$, $G$
$\tilde{X}^{t}$ capture the distributions of the
entries of $\mathbf{x}$, entries of $\mathbf{v}^{t}$ (up to a factor $\tau_{t}$), the prior information $g_n$,
and entries of $\mathbf{\tilde{x}}^{t}$, respectively, with $\mathbb{E}\big[|\eta_t(\tilde{X}^{t},G)-X|^2\big]$ characterizing the per-coordinate MSE of the estimate of
$\mathbf{x}$ at iteration $t$.

\subsection{AMP for MMV problem}
\subsubsection{AMP with vector denoiser} One extension of the AMP algorithm to solve the MMV problem in \eqref{eq.signal.multiple} is proposed in \cite{Kim2011}, which employs a vector denoiser that operates on each row vector of the matched filtered output:
\begin{align}
\mathbf{X}^{t+1}&=\eta(\mathbf{S}^{*}\mathbf{Z}^{t}+\mathbf{X}^{t},\mathbf{g},t), \label{eq.vamp1}\\
\mathbf{Z}^{t+1} &= \mathbf{Y}-\mathbf{S}\mathbf{X}^{t+1}+\frac{N}{L}\mathbf{Z}^{t}\langle\eta^{\prime}(\mathbf{S}^{*}\mathbf{Z}^{t}+\mathbf{X}^{t},\mathbf{g},t)\rangle, \label{eq.vamp2}
\end{align}
where $\eta(\cdot,\mathbf{g},t)\triangleq [\eta_t(\cdot,g_1),\cdots,\eta_{t}(\cdot,g_N)]^{T}$ with $\eta_t(\cdot,g_n): \mathbb{C}^{1\times M}\rightarrow \mathbb{C}^{1\times M}$ is a vector denoiser that operates on the $n$th row vector of $\mathbf{S}^{*}\mathbf{Z}^{t}+\mathbf{X}^{t}$, and the other notations are similar to those used in \eqref{eq.amp1} and \eqref{eq.amp2}. The state evolution of the AMP algorithm for the MMV problem also has a similar form as
\begin{align}\label{eq.stateevo.multiple}
\mathbf{\Sigma}_{t+1} = \sigma_w^2\mathbf{I}+\frac{N}{L}\mathbb{E}\big[\mathbf{D}^{t}(\mathbf{D}^{t})^{*}\big],
\end{align}
where $\mathbf{D}^{t}\triangleq \left(\eta_t(\mathbf{R}+\mathbf{U}^{t},G)-\mathbf{R}\right)^{T}\in \mathbb{C}^{M\times 1}$ , with random vector $\mathbf{R} $ following $p_{\mathbf{R}|G}$ and random vector $\mathbf{U}^{t}$ following
$\mathcal{CN}(\mathbf{0},\mathbf{\Sigma}_{t})$. The expectation is taken over $\mathbf{R}$, $\mathbf{U}^{t}$, and $G$. To minimize the estimation error at each iteration, we can also design the vector denoiser $\eta_t(\cdot,\cdot)$ via the Bayesian approach.

\subsubsection{Parallel AMP-MMV} A different extension of the AMP algorithm for
dealing with the MMV problem is the parallel AMP-MMV algorithm proposed in
\cite{Ziniel2013}. The basic idea is to solve the MMV problem iteratively by
using multiple parallel AMP-SMVs then exchanging soft information between them.
Parallelization allows distributed implementation of the algorithm, which can
be computationally advantageous, especially when the number of antennas is
large.

The outline of the parallel AMP-MMV algorithm is illustrated in Algorithm~\ref{Alg:1} which
operates on a per-antenna basis, i.e., on the columns of $\mathbf{X}$ and $\mathbf{Z}$, denoted as $\mathbf{c}_m$ and $\mathbf{z}_m$ respectively, and where
$\eta(\cdot,\mathbf{g},t,i,m)\triangleq [\eta_{t,i,m}(\cdot,g_1),\cdots,\eta_{t,i,m}(\cdot,g_N)]^{T}$ is the denoiser used for the $m$th antenna in the iteration $(t,i)$. Note that here we add index  $i$ and $m$ in the notation of denoiser, $\eta_{t,i,m}(\cdot,g_n)$, to indicate the index of outer iteration and the index of SMV stage, respectively. In the first phase which is called the (into)-phase, the messages $\harpoonl{\pi}_{nm} $, are calculated and passed to the $m$th AMP-SMV stage. These messages convey the current belief about the probability of being active for each user. In the first iteration, we have $\harpoonl{\pi}_{nm} = \lambda, \forall{n,m}$, since no further information is available. In the next phase, which is called the (within)-phase, the conventional AMP algorithm is applied to the received signal of each antenna. Note that the denoiser in AMP algorithm is a function of the current belief about the activity of the users which is obtained based on the information sharing between all $M$ AMP-SMV stages. Finally, in the (out)-phase, the estimate of channel gains is used to refine the belief about the activity of the users.

\begin{algorithm}[t]
\caption{Parallel AMP-MMV Method \cite{Ziniel2013}}
\label{Alg:1}
\begin{algorithmic}[1]
\State Initialize $\harpoonr{\pi}_{nm}  = 0.5, \quad \forall n,m$.
\For{$i = 1 ~\mathrm{to}~ I$}
\Statex Execute the (into)-phase:
\State
$\harpoonl{\pi}_{nm} = \frac{\lambda \prod_{m^{\prime} \not = m} \harpoonr{\pi}_{nm^{\prime}}} {(1-\lambda)\prod_{m^{\prime} \not = m} (1- \harpoonr{\pi}_{nm^{\prime}}) + \lambda \prod_{m^{\prime} \not = m} \harpoonr{\pi}_{nm^{\prime}}}, \, \forall n,m
$
\Statex Execute the (within)-phase:
\For{$m = 1 ~\mathrm{to}~ M$}
\State Initialize $\mathbf{c}^{0}_m = \mathbf{0}$, $\mathbf{z}^{0}_m = \mathbf{y}_m$.
 \For{$t= 0 ~\mathrm{to}~ T$}
\State $\mathbf{c}^{t+1}_m=\eta(\mathbf{S}^{*}\mathbf{z}^{t}_m+\mathbf{c}_m^{t},\mathbf{g},t,i,m)$,
\State $\mathbf{z}^{t+1}_m= \mathbf{y}_m-\mathbf{S}\mathbf{c}^{t+1}_m+\frac{N}{L}\mathbf{z}_{m}^{t}\langle\eta^{\prime}(\mathbf{S}^{*}\mathbf{z}^{t}_m+\mathbf{c}^{t}_m,\mathbf{g},t,i,m)\rangle $,
 \EndFor
 \EndFor
\Statex Execute the (out)-phase:
\State Calculate $\harpoonr{\pi}_{nm}, \forall n,m$,
the probability of user $n$ being active based on the decision at the $m$th AMP-SMV stage.
\EndFor
\end{algorithmic}
\end{algorithm}

\section{User Activity Detection: Single-Antenna Case}
\label{sec.single}

A main point of this paper is that exploiting the statistics of the wireless
channel can significantly enhance detector performance. This section proposes
an MMSE denoiser design for the AMP algorithm for the wireless massive
connectivity problem that specifically takes wireless channel characteristics
into consideration in the single-antenna case.
Two scenarios are considered: the large-scale fading $g_n$ of each user
is either directly available or only its statistics is available at the BS.
This section further studies the optimal detection strategy, and analyzes the
probabilities of false alarm and missed detection by using the state evolution
of the AMP algorithm.

\subsection{MMSE Denoiser for AMP Algorithm}
In the scenario where only the statistics about the large-scale fading is known at the BS, the distributions of the channel coefficients $p_{H}(h_n)$ are
independent and identical for all devices. In the scenario where the devices are stationary and
their path-loss and shadowing coefficients can be
estimated and thus the exact large-scale fading is known at the BS, the distributions of
the channel coefficients are of the form $p_{H|G}(h_n|g_n)$, which are complex Gaussian
with variance parameterized by $g_n$, and are independent but not identical
across the devices.
To derive the MMSE denoisers via the Baysian approach for both
cases, we first characterize the distributions $p_{G}(g_n)$,
$p_{H}(h_n)$ and $p_{H|G}(h_n|g_n)$ as follows.

\begin{proposition}\label{P.Prop1}
Consider a circular wireless cellular coverage area of radius $R$ with BS at the center and uniformly distributed devices where the channels between the BS and the devices are modeled with large-scale fading $g_n$ with parameters $\alpha, \beta$ and shadowing fading parameter $\sigma_{\mathrm{SF}}$ as defined in the system model. Then, $g_n$ follows a distribution as
\begin{equation}\label{eq.distr.of.lsf}
p_{G}(g_n) = ag_n^{-\gamma}Q(g_n),
\end{equation}
where $Q(g_n)\triangleq \int_{(b\ln g_n+c)}^{\infty}\exp(-s^2)ds$, $\gamma \triangleq 40/\beta+1$, and $a, b$, and $c$ are constants depending on parameters $\alpha, \beta, \sigma_{\mathrm{SF}}$ and $R$ as
\begin{align}\label{A.constant}
a&=\frac{40}{R^2\beta\sqrt{\pi}}\exp\left(\frac{2(\ln{10})^2\sigma_{\mathrm{SF}}^2}{\beta^2}-\frac{2\ln(10)\alpha}{\beta}\right),\nonumber\\
b&=\frac{-10\sqrt{2}}{(\ln10)\sigma_{\mathrm{SF}}},\quad c=\frac{-\alpha-\beta\log_{10}(R)}{\sqrt{2}\sigma_{\mathrm{SF}}}-\frac{20}{\beta b}.\nonumber
\end{align}
\end{proposition}
\begin{proof}
See Appendix~\ref{A:probG}.
\end{proof}

\begin{proposition}\label{P.Prop2}
Denote $h_n$ as the channel coefficient which contains both the large-scale fading $g_n$ and Rayleigh fading. If only $p_{G}(g_n)$ is known at the BS, the pdf of $h_n$ is given by
\begin{equation}\label{eq.distr.of.channel}
p_{H}(h_n) = \int_{0}^{\infty}\frac{a}{\pi}g_n^{-\gamma-2}Q(g_n)\exp \left(\frac{-|h_n|^2}{g_n^2}\right)dg_n.
\end{equation}
If $g_n$ is known at the BS, the pdf of $h_n$ given $g_n$ is
\begin{equation}\label{eq.distr.of.channel.lsf}
p_{H|G}(h_n|g_n) = \frac{1}{\pi g_n^2}\exp\left(\frac{-|h_n|^2}{g_n^2}\right).
\end{equation}
\end{proposition}
\begin{proof}
See Appendix~\ref{A:probH}.
\end{proof}

Note that for the first scenario, the channel distribution
\eqref{eq.distr.of.channel} only depends on a few parameters such as the
path-loss exponent in the path-loss model and the standard deviation in the
shadowing model, which are assumed to be known and can be estimated in
practice. For the second scenario, the channel distribution
\eqref{eq.distr.of.channel.lsf} is just a Rayleigh fading model parameterized
by the large-scale fading. The large-scale fading information can be
obtained by tracking the estimated
channel over a reasonable period. This second scenario is applicable to the
case where the users are mostly stationary, so the large-scale fading changes
only slowly over time. 
It is worth noting that although this paper restricts attention to the Rayleigh
fading model, the approach developed here is equally applicable for Rician or
any other statistical channel model.


In the following, we design the MMSE denoisers for the AMP algorithm by exploiting $p_{H}(h_n)$ and $p_{H|G}(h_n|g_n)$.

\subsubsection{With Statistical Knowledge of Large-Scale Fading Only}
Since $g_n$ is unknown and only the distribution $p_{G}(g_n)$ is available at the BS, the denoiser $\eta_t(\cdot,g_n)$ reduces to $\eta_t(\cdot)$, which indicates that the denoiser for each entry of the matched filtered output is the same. By using \eqref{eq.distr.of.channel}, the pdf of the entries of $\mathbf{x}$ can be expressed as
\begin{equation}\label{eq.distr.of.x}
p_X(x_n)=(1-\lambda)\delta_{0}+\int_{0}^{\infty}\frac{\exp (-|x_n|^2g_n^{-2})}{\pi g_{n}^{\gamma+2}/(a\lambda Q(g_n))} dg_n.
\end{equation}
The MMSE denoiser is given by the conditional expectation, i.e.,  $\eta_t(\tilde{x}^{t}_n)=\mathbb{E}[X|\tilde{X}^{t}=\tilde{x}^{t}_n]$ where random variable $\tilde{X}^{t}=X+\tau_tV$, and $\tilde{x}^{t}_n$ is a realization of $\tilde{X}^{t}$. Note that the denoiser $\eta_t(\tilde{x}^{t}_n)$ depends on $t$ through $\tau_t$. The expression of the conditional expectation is given in the following proposition.

\begin{proposition}\label{P.Prop3}
Based on the pdf of $x_n$ in \eqref{eq.distr.of.x} and the signal-plus-noise model $\tilde{X}^{t}=X+\tau_tV$ at each iteration in AMP, the conditional expectation of $X$ given $\tilde{X}^{t}=\tilde{x}^{t}_n$ is given by
\begin{align}\label{eq.mmseEst}
\mathbb{E}[X|\tilde{X}^{t}=\tilde{x}^{t}_n]=\tilde{x}^{t}_n\frac{\nu_{1}(|\tilde{x}^{t}_n|^2)}{\xi_{1}(|\tilde{x}^{t}_n|^2)},
\end{align}
where functions $\nu_i(s)$ and $\xi_i(s)$ are defined as
\begin{align}
\nu_{i}(s)\triangleq&\int_{0}^{\infty}\frac{g_n^{2-\gamma}Q(g_n)}{(g_n^2+\tau_{t}^2)^{i+1}}\exp\left(\frac{-s}{g_n^2+\tau_{t}^2}\right)dg_n, \label{A.func.nu}\\
\xi_{i}(s)\triangleq& \frac{1-\lambda}{\lambda
a\tau_t^{2i}}\exp \left(\frac{-s}{\tau_t^{2}}\right)\nonumber\\
&+\int_{0}^{\infty}\frac{g_n^{-\gamma}Q(g_n)}{(g_n^2+\tau_{t}^2)^i}\exp\left({\frac{-s}{g_n^2+\tau_{t}^2}}\right)dg_n.\label{eq.xi}
\end{align}
\end{proposition}
\begin{proof}
See Appendix~\ref{A:mmse}.
\end{proof}

Note that to implement $\eta_{t}(\tilde{x}^{t}_n)$ at each iteration, the value of $\tau_t$ is needed. In practice, an empirical estimate
$\hat{\tau}_t=\frac{1}{\sqrt{L}}\|\mathbf{z}^t\|_2$, where $\|\cdot\|_2$ denotes the $\ell_2$ norm,
can be used \cite{Montanari2012}.
Although
$\eta_{t}(\tilde{x}^{t}_n)$ is in a complicated form, we
note that it can be pre-computed and stored as table lookup, so it does not add to run-time
complexity. To gain some intuition, we illustrate the shape of
the MMSE denoiser as compared to the widely used soft thresholding
denoiser in Fig.~\ref{fig.denoiser}. We observe that the MMSE
denoiser plays a role similar to the soft thresholding denoiser,
shrinking the input towards the origin, especially when the input is
small, thereby promoting sparsity.

\begin{figure}
\centerline{\epsfig{figure=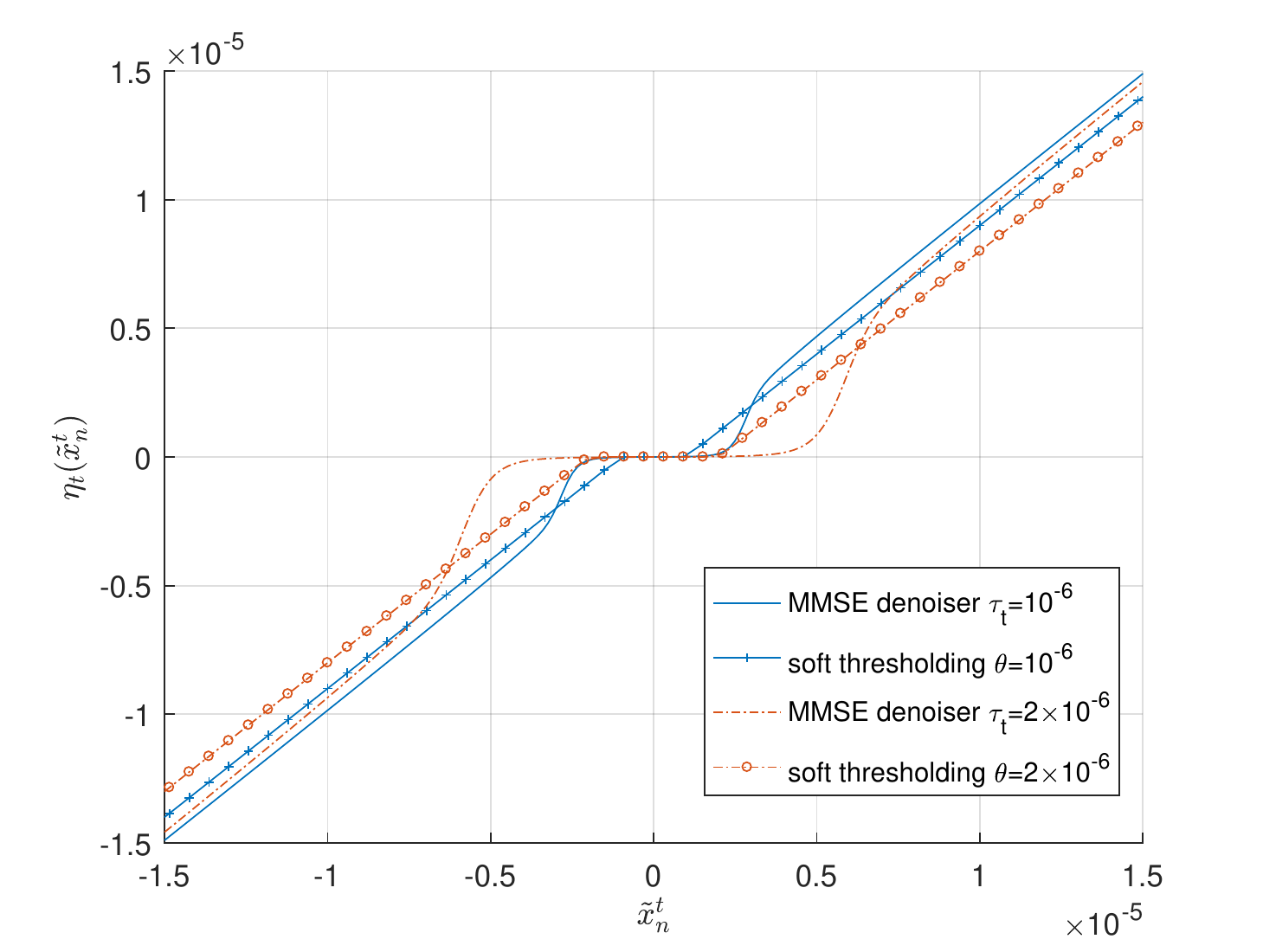,width=9cm}}
\caption{MMSE denoiser vs. soft thresholding denoiser \cite{Maleki2013} $\eta_t^{\mathrm{soft}}(\tilde{x}^{t}_n)\triangleq(\tilde{x}^{t}_n-\frac{\theta\tilde{x}^{t}_n}{|\tilde{x}^{t}_n|})\mathbb{I}(|\tilde{x}^{t}_n|>\theta)$, where $\mathbb{I}(\cdot)$ is the indicator function.}
\label{fig.denoiser}
\end{figure}

\subsubsection{With Exact Knowledge of Large-Scale Fading}
When $g_n$ is available at the BS, we substitute \eqref{eq.distr.of.channel.lsf} into \eqref{eq.distr}, and the pdf of the entries of $\mathbf{x}$ is simplified to Bernoulli-Gaussian as
\begin{equation}\label{eq.distr.of.x.lsf}
p_{X|G}(x_n|g_n)=(1-\lambda)\delta_{0}+ \frac{\lambda}{\pi g_n^2}\exp\left(\frac{-|x_n|^2}{g_n^2}\right).
\end{equation}
The MMSE denoiser is given by $\eta_t(\tilde{x}^{t}_n,g_n)=\mathbb{E}[X|\tilde{X}^{t}=\tilde{x}^{t}_n,G=g_n]$, where the conditional expectation is \cite{Schniter2010}
\begin{equation}\label{eq.mmseEst.lsf}
\mathbb{E}[X|\tilde{X}^{t}=\tilde{x}^{t}_n,G=g_n]=\frac{g_n^2(g_n^2+\tau_t^2)^{-1}\tilde{x}^{t}_n}{1+\frac{1-\lambda}{\lambda}\frac{g_n^2+\tau_t^2}{\tau_t^2}\exp\left(-\Delta|\tilde{x}^{t}_n|^2\right)},
\end{equation}
where
\begin{equation}
\Delta\triangleq \tau_t^{-2}-(g_n^2+\tau_t^2)^{-1}.
\label{eq:Delta}
\end{equation}
Compared with the MMSE denoiser in \eqref{eq.mmseEst}, we add $g_n$ to the left hand side of \eqref{eq.mmseEst.lsf} to emphasize the dependency on prior information $g_n$.

\subsection{User Activity Detection}
After the AMP algorithm has converged, we employ the likelihood ratio
test to perform user activity detection. For the hypothesis testing problem
\begin{align}
\left\{
\begin{aligned}
H_0 &:  X=0, \,\,\mathrm{inactive\,\, user},\\
H_1 &: X\neq0, \,\,\mathrm{active\,\, user};
\end{aligned}
\right.
\end{align}
the optimal decision rule is given by
\begin{align}\label{eq.decision}
\operatorname{LLR}=\log\left(\frac{p_{\tilde{X}^{t}|X}(\tilde{x}^{t}_n|X\neq0)}{p_{\tilde{X}^{t}|X}(\tilde{x}^{t}_n|X=0)}\right)\mathop{\lessgtr}\limits_{H_1}^{H_0} l_n,
\end{align}
where $\operatorname{LLR}$ denotes the log-likelihood ratio, and $l_n$ denotes the
decision threshold typically determined by a cost function. The performance metrics of interest are the probability of missed detection $P_M$, defined as the probability that a device is active but the detector declare the null hypothesis $H_0$, and the probability of false alarm, $P_F$, defined as the probability that a device is inactive, but the detector declare it to be active. We consider the threshold for two cases depending on whether the large-scale fading coefficient $g_n$ is available at the BS or not.

\subsubsection{With Statistical Knowledge of Large-Scale Fading Only}
We first derive the likelihood probabilities
in the following.
\begin{proposition}\label{P.Prop4}
Suppose that $X$ follows \eqref{eq.distr.of.x}, and $V$ follows complex
Gaussian distribution with zero mean and unit variance, the likelihood
of $\tilde{X}^{t}= X+\tau_{t}V$ given $X=0$ or $X\neq 0$ is given by
\begin{align}
p_{\tilde{X}^{t}|X}(\tilde{x}^{t}_n|X=0) &= \frac{1}{\pi\tau_t^2}\exp\left(\frac{-|\tilde{x}^{t}_n|^2}{\tau_{t}^2}\right), \label{eq.likelihood1}\\
p_{\tilde{X}^{t}|X}(\tilde{x}^{t}_n|X\neq 0) &= \int_{0}^{\infty}\frac{ag_n^{-\gamma}Q(g_n)}{\pi(g_n^2+\tau_t^2)}\exp\left(\frac{-|\tilde{x}^{t}_n|^2}{g_n^2+\tau_{t}^2}\right)dg_n.\label{eq.likelihood2}
\end{align}
\end{proposition}
\begin{proof}
See Appendix \ref{A:probHN}.
\end{proof}
Based on \eqref{eq.likelihood1} and \eqref{eq.likelihood2}, the log-likelihood ratio is given as
\begin{align}
\operatorname{LLR} = \log\int_{0}^{\infty}\frac{a\tau_t^2 g_n^{-\gamma}}{g_n^2+\tau_t^2}Q(g_n)\exp(|\tilde{x}^{t}_n|^2\Delta)dg_n,
\end{align}
where $\Delta$ is defined in \eqref{eq:Delta}. By observing that $\operatorname{LLR}$ is monotonic in $|\tilde{x}^{t}_n|$, we can simplify the decision rule in \eqref{eq.decision} as
$|\tilde{x}^{t}_n| \mathop{\lessgtr}\limits_{H_1}^{H_0} l_n$, indicating that user activity detection can be performed based on the magnitude of
$\tilde{x}^{t}_n$ only.

Based on the likelihood probabilities and the threshold $l_n$, the probabilities of false alarm and missed
detection can be characterized as follows
\begin{align}
P_F &= \int_{|\tilde{x}^{t}_n|>l_n}p_{\tilde{X}^{t}|X}(\tilde{x}^{t}_n|X=0)d\tilde{x}^{t}_n=\exp\left(\frac{-l_n^2}{\tau_t^2}\right),\label{eq.pf}\\
P_M &= \int_{|\tilde{x}^{t}_n|<l_n}p_{\tilde{X}^{t}|X}(\tilde{x}^{t}_n|X\neq0) d\tilde{x}^{t}_n,\label{eq.pm}
\end{align}
where \eqref{eq.pf} is simplified by using \eqref{eq.likelihood1}. Note that since only statistical information of the large-scale fading is known at the BS, $P_F$ and $P_M$ are the averaged false alarm and missed detection probabilities which do not depend on $g_n$ .

\subsubsection{With Exact Knowledge of Large-Scale Fading}
When $g_n$ is known at the BS, the distribution of $X$ is simplified to Bernoulli-Gaussian. The likelihood probabilities become
\begin{align}
p_{\tilde{X}^{t}|X,G}(\tilde{x}_n^{t}|X=0,G=g_n) &= \frac{\exp\left(-|\tilde{x}_n|^2\tau_{t}^{-2}\right)}{\pi\tau_t^2}, \label{eq.likelihood.kls1}\\
p_{\tilde{X}^{t}|X,G}(\tilde{x}_n^{t}|X\neq0,G=g_n) &= \frac{\exp\left(-|\tilde{x}_n|^2(g_n^2+\tau_{t}^2)^{-1}\right)}{\pi(\tau_t^2+g_n^2)}.\label{eq.likelihood.kls2}
\end{align}
The log-likelihood ratio is then given as
\begin{align}
\operatorname{LLR}(g_n) = \log\left(\frac{\tau_t^2}{g_n^2+\tau_t^2}\exp(|\tilde{x}^{t}_n|^2\Delta)\right),
\end{align}
where the notation $\operatorname{LLR}(g_n)$ emphasizes the dependency on the prior information $g_n$.
Similar to the case where only the statistics of $g_n$ is known, $\operatorname{LLR}$ here is also monotonic in $|\tilde{x}^{t}_n|$, which means that the user activity detection can be performed based on $|\tilde{x}^{t}_n|$ only.

We also use $l_n$ to denote the threshold in the detection. Based on \eqref{eq.likelihood.kls1} and \eqref{eq.likelihood.kls2}, the probabilities of false alarm and missed detection probability are given as follows
\begin{align}
P_F(g_n) = & \int_{|\tilde{x}^{t}_n|>l_n} p_{\tilde{X}^{t}|X,G}(\tilde{x}^{t}_n|X=0,G=g_n)d\tilde{x}^{t}_n\nonumber\\
=&\exp\left(-l_n^2\tau_t^{-2}\right),\label{eq.pf.lsf}\\
P_M(g_n) = & \int_{|\tilde{x}^{t}_n|<l_n}p_{\tilde{X}^{t}|X,G}(\tilde{x}^{t}_n|X\neq 0,G=g_n)d\tilde{x}^{t}_n\nonumber\\
=&\,1-\exp\left(-l_n^2(g_n^2+\tau_t^2)^{-1}\right), \label{eq.pm.lsf}
\end{align}
where we use the notation $P_F(g_n)$ and $P_M(g_n)$ to indicate the prior known $g_n$.
Note that the false alarm probability in \eqref{eq.pf.lsf} has the form as that in \eqref{eq.pf} even through the value of $\tau_t$ may be different due to different denoisers.

A natural question then arises: how to design the threshold $l_n$ as a function of the known large-scale fading $g_n$? In theory, we can treat each user separately, i.e., set the thresholding value of each user separately according to its own cost function. For example, if a specific target false alarm probability is needed for user $n$, we can design its thresholding parameter, $l_n$, using the expression in \eqref{eq.pf.lsf}. In order to bring fairness, this paper considers a common target false alarm probability for all users. Under this condition, all users share the same thresholding parameter, i.e., $l_n = l, \forall n$, since the expression of $P_F(g_n)$ in \eqref{eq.pf.lsf} does not depend on $g_n$. In such a case, different users may have different probabilities of missed detection depending on their large-scale fading $g_n$. To measure the performance of the detector for the entire system, we employ the average probability of missed detection as
\begin{align}
\label{common_l_PM}
\overline{P_M} &= \frac{1}{N}\sum_{n=1}^N \left(1-\exp\left(\frac{-l^2}{\tau_t^{2}+g_n^2}\right) \right) \nonumber\\
&\rightarrow
\int p_G(g) \left(1-\exp\left(\frac{-l^2}{\tau_t^{2}+g^2}\right) \right) dg,\, \mathrm{as} \,\, N \rightarrow \infty,
\end{align}
where the distribution $P_G(g)$ is given in \eqref{eq.distr.of.lsf}. When $N$ is large, once $\tau_t$ is given, the averaged performance only depends on the statistics of the large-scale fading $g_n$.

\subsection{State Evolution Analysis}
We have characterized the probabilities of false alarm $P_F$ and missed detection $P_M$ for user activity detection in \eqref{eq.pf}, \eqref{eq.pm} and \eqref{eq.pf.lsf}, \eqref{eq.pm.lsf}, but the parameter $\tau_t$ that represents the standard deviation of the residual noise still needs to be determined. As AMP proceeds, $\tau_t$ converges to
$\tau_{\infty}$.
To compute $\tau_t$, we use the state evolution
\eqref{eq.stateevo}, where $\mathbb{E}\big[|\eta_t(\tilde{X}^{t},G)-X|^2\big]$ in \eqref{eq.stateevo} can be interpreted as the MSE of the denoiser. Note that for the MMSE denoiser, MSE can also be expressed as $\mathbb{E}\big[|\eta_t(\tilde{X}^{t},G)-X|^2\big]=\mathbb{E}\big[\operatorname{Var}(X|\tilde{X}^{t},G)\big]$, where $\operatorname{Var}(X|\tilde{X}^{t},G)$ is the conditional variance of $X$ given $\tilde{X}^{t}$ and $G$, and the expectation is taken over both $\tilde{X}^{t}$ and $G$. (Note that we drop $G$ if the large-scale fading coefficient is unknown.) By using conditional variance, we characterize the MSE of the designed denoisers in the following propositions.

\begin{proposition}\label{P.Prop5}
The MSE of the denoiser for the case where only the statistics of $g_n$ is known to the BS is given by
\begin{align}\label{eq.mseofmmse}
\mathrm{MSE}(\tau_t)=&\int_{0}^{\infty}\frac{aQ(g_n)}{g_n^{\gamma}}\cdot\frac{\lambda g_n^2\tau_t^2}{g_n^2+\tau_t^2}dg_n\nonumber\\
&+\int_{0}^{\infty}a\lambda s\left(\mu_{1}(s)-\nu^{2}_{1}(s)\xi_{1}^{-1}(s)\right)ds,
\end{align}
where functions $\nu_{i}(s)$ and $\xi_{i}(s)$ are defined in \eqref{A.func.nu} and \eqref{eq.xi}, respectively, and function $\mu_{i}(s)$ is defined as
\begin{align}
\mu_{i}(s)\triangleq&\int_{0}^{\infty}\frac{g_n^{4-\gamma}Q(g_n)}{(g_n^2+\tau_{t}^2)^{i+2}}\exp\left(\frac{-s}{g_n^2+\tau_{t}^2}\right)dg_n. \label{A.func.mu}
\end{align}
\end{proposition}
\begin{proof}
See Appendix~\ref{A:MSE}.
\end{proof}
It is worth noting that $\lambda g_n^2\tau_t^2(g_n^2+\tau_t^2)^{-1}$ in the first term of the right hand side of \eqref{eq.mseofmmse} corresponds to the MSE of the estimate of $x_{n}$ if the large-scale fading coefficient $g_n$ as well as the user activity is assumed to be a priori known, and the integral of $g_n$ corresponds to the averaging over all possible $g_n$. The second term then represents the cost of unknown $g_n$ and unknown user activity in reality. Similarly, for the case where $g_n$ is exactly known, the MSE can be characterized as follows.

\begin{proposition}\label{P.Prop6}
The MSE of the denoiser for the case where $g_n$ is known exactly at the BS is
\begin{align}\label{eq.mseofmmse.lsf}
\mathrm{MSE}(\tau_t)=&\int_{0}^{\infty}\frac{aQ(g_n)}{g_n^{\gamma}}\cdot\frac{\lambda g_n^2\tau_t^2}{g_n^2+\tau_t^2}dg_n\nonumber\\
&+\int_{0}^{\infty}\frac{a\lambda Q(g_n)g_n^4}{g_n^{\gamma}(g_n^2+\tau_t^2)}\left(1-\varphi_1(g_n^2\tau_t^{-2})\right)dg_n,
\end{align}
where function $\varphi_{i}(s)$ of $s$ is defined as
\begin{align}\label{A.func.varphi}
\varphi_{i}(s)\triangleq\int_{0}^{\infty}\frac{t^i\exp(-t)}{1+(1-\lambda)(1+s)^i\exp(-st)/\lambda}dt.
\end{align}
\end{proposition}
\begin{proof}
See Appendix \ref{A.MSE.lsf}.
\end{proof}

We also observe from \eqref{eq.mseofmmse.lsf} that the first term in the right hand side corresponds to the averaged MSE if the user activity is assumed to be known, and the second term corresponds the extra error brought by unknown user activity.

Based on the expressions of MSE in \eqref{eq.mseofmmse} and \eqref{eq.mseofmmse.lsf}, the state evolution in \eqref{eq.stateevo} can be expressed as
\begin{align}
\tau_{t+1}^2 =\sigma_w^2+\frac{N}{L}\mathrm{MSE}(\tau_{t}),
\end{align}
based on which $P_F$ and $P_M$ can be evaluated according to \eqref{eq.pf}, \eqref{eq.pm}, and \eqref{eq.pf.lsf}, \eqref{eq.pm.lsf}, as functions of the iteration number. As the AMP algorithm converges, $\tau_t$ converges to the fixed point $\tau_{\infty}$ of the above equation.

Now we compare the resulting MSEs in these two cases. According to the decomposition of variance, we have
\begin{align}
\mathbb{E}\big[\operatorname{Var}\big(X|\tilde{X}^{t}\big) \big]& = \mathbb{E}\big[\operatorname{Var}\big(X|\tilde{X}^{t},G\big)\big] \nonumber\\
&\quad+ \mathbb{E}\big[\operatorname{Var}\big(\mathbb{E}[X|\tilde{X}^{t},G]\big|\tilde{X}^{t}\big)\big]\nonumber\\
&\geq\mathbb{E}\big[\operatorname{Var}\big(X|\tilde{X}^{t},G\big)\big],
\end{align}
which indicates that knowing the large-scale fading can help to improve the estimation on $X$ given $\tilde{X}^{t}$. However, the simulation results in Section~\ref{sec.numerical} show that surprisingly for the model of the large-scale fading considered in this paper, the performance improvement is actually minor, indicating that knowing the large-scale fading does not help to get a much better estimation. Knowing the exact value of $g_n$ is not crucial in user activity detection and the statistical information of $g_n$ is sufficient for device detection.

\section{User Activity Detection: Multiple-Antenna Case}
\label{sec.multiple}
This section designs the AMP algorithms that account for wireless channel
propagation for the massive connectivity problem in the multiple-antenna case.
As mentioned earlier, two different AMP algorithms can be used for the MMV
problem: the AMP with a vector denoiser operating on each row of the
input matrix, or the parallel AMP-MMV that divides the MMV problem into
parallel SMV problems and iteratively solves the SMV problem on each antenna
separately with soft information exchange between the antennas. The AMP with
vector denoiser admits a state evolution, which allows an easier
characterization of its performance, whereas AMP-MMV can be implemented in a
distributed way which is helpful for reducing the running time of the
algorithm, especially when the BS is equipped with large antenna arrays.
\subsection{User Activity Detection by AMP with Vector Denoiser}
As in the scenario with single antenna, we consider both the cases where only the statistical knowledge or the exactly knowledge of the large-scale fading is known at the BS. To design the denoisers, we first characterize the pdfs of the row vectors of $\mathbf{X}$ in the following.

\begin{proposition}\label{P.Prop7}
Denote $\mathbf{r}_n$ as the row vector of $\mathbf{X}$. If only $p_{G}(g_n)$ is known at the BS, the pdf of $\mathbf{r}_n$ is given by
\begin{align}\label{eq.vamp.nolsf.distri}
p_{\mathbf{R}}(\mathbf{r}_n)=(1-\lambda)\delta_{\mathbf{0}}+\int_{0}^{\infty}\frac{\exp \left(-\|\mathbf{r}_n\|_2^2g_n^{-2}\right)}{\pi^{M} g_n^{\gamma+2M}/(a\lambda Q(g_n))} dg_n.
\end{align}
If $g_n$ is known, the pdf of $\mathbf{r}_n$ is Bernoulli-Gaussian as
\begin{align}\label{eq.vamp.z}
p_{\mathbf{R}|G}(\mathbf{r}_n|g_n)=(1-\lambda)\delta_{\mathbf{0}}+\frac{\lambda\exp(-\|\mathbf{r}_n\|_2^2g_n^{-2})}{(\pi g_n^2)^M}.
\end{align}
\end{proposition}
\begin{proof}
The results are extensions of \eqref{eq.distr.of.x} and \eqref{eq.distr.of.x.lsf} by considering  multivariate random variables.
\end{proof}
Given $\mathbf{\tilde{R}}^{t}=\mathbf{R}+\mathbf{U}^{t}_n$ with $\mathbf{U}^t$ following complex Gaussian distribution with zero mean and covariance $\mathbf{\Sigma}_{t}$, the MMSE denoisers $\eta_t({\mathbf{\tilde{r}}}^t_n)$ and $\eta_t({\mathbf{\tilde{r}}}^t_n,g_n)$ for both cases are given by the conditional expectation in the following.

\begin{proposition}\label{P.Prop8}
If only $p_{G}(g_n)$ is known at the BS, the conditional expectation of $\mathbf{R}$ given $\mathbf{\tilde{R}}^{t}=\mathbf{\tilde{r}}_{n}^{t}$ is
\begin{align}\label{eq.vamp.nolsf.mmse}
\mathbb{E}[\mathbf{R}|\mathbf{\tilde{R}}^{t}=\mathbf{\tilde{r}}_{n}^{t}]=\frac{\int_{0}^{\infty}Q(g_n)\psi_a(g_{n})(g_{n}^{-2}\mathbf{\Sigma}_t+\mathbf{I})^{-1}\mathbf{\tilde{r}}^t_ndg_n}
{\psi_c(g_n)+\int_{0}^{\infty}Q(g_n)\psi_b(g_n)dg_n},
\end{align}
where $\psi_a(g_n)$, $\psi_b(g_n)$ and $\psi_c(g_n)$ are defined as follows
\begin{align}
\psi_a(g_n) \triangleq & \frac{\exp\left(-\mathbf{\tilde{r}}^t_n\left(\mathbf{\Sigma}_t^{-1}-(\mathbf{\Sigma}_t+g_n^{-2}\mathbf{\Sigma}_t^2)^{-1}\right)(\mathbf{\tilde{r}}^t_n)^{*}\right)}{g_n^{\gamma}|\mathbf{\Sigma}_t+g_n^2\mathbf{I}|},\\
\psi_b(g_n) \triangleq & \frac{\exp\left(-\mathbf{\tilde{r}}^t_n\left(g_n^{-2}\mathbf{I}-(g_n^{2}\mathbf{I}+g_n^4\mathbf{\Sigma}_t^{-1})^{-1}\right)(\mathbf{\tilde{r}}^t_n)^{*}\right)}{g_n^{\gamma}|\mathbf{\Sigma}_t+g_n^2\mathbf{I}|},\\
\psi_c(g_n) \triangleq &(1-\lambda)(a\lambda)^{-1}
\exp\left(-\mathbf{\tilde{r}}^t_n\mathbf{\Sigma}_{t}^{-1}(\mathbf{\tilde{r}}^t_n)^{*}\right)|\mathbf{\Sigma}_t|^{-1}.
\end{align}
If $g_n$ is known at the BS, the conditional expectation is
\begin{align}\label{eq.vamp.mmse}
\mathbb{E}[\mathbf{R}|\mathbf{\tilde{R}}^{t}=\mathbf{\tilde{r}}_{n}^{t},G=g_n]=\frac{( g_n^{-2}\mathbf{\Sigma}_t +\mathbf{I})^{-1}\mathbf{\tilde{r}}^{t}_n}{1+\frac{1-\lambda}{\lambda}|g_n^2\mathbf{I}+\mathbf{\Sigma}_t|\psi_d(g_n)},
\end{align}
where $\psi_d(g_n)$ is defined as follows
\begin{align}
\psi_d(g_n)\triangleq\exp\left(-\mathbf{\tilde{r}}^{t}_n\left(\mathbf{\Sigma}_t^{-1}-(\mathbf{\Sigma}_t+g_n^2\mathbf{I})^{-1}\right)(\mathbf{\tilde{r}}^{t}_n)^{*}\right)|\mathbf{\Sigma}_t|^{-1}.
\end{align}
\end{proposition}
\begin{proof}
See Appendix~\ref{A.mmse.v}.
\end{proof}

The covariance matrix $\mathbf{\Sigma}_{t}$ in both \eqref{eq.vamp.nolsf.mmse} and \eqref{eq.vamp.mmse} is tracked via the state evolution \eqref{eq.stateevo.multiple}, and $\mathbf{\Sigma}_{t}$ can be further simplified by the following proposition.

\begin{proposition}\label{prop.evo}
Based on the pdfs in \eqref{eq.vamp.nolsf.distri} and \eqref{eq.vamp.z} and the
state evolution \eqref{eq.stateevo.multiple}, if the initial covariance matrix
$\mathbf{\Sigma}_{0}$ is a diagonal matrix with identical diagonal entries,
i.e., $\mathbf{\Sigma}_{0}=\tau_0^2\mathbf{I}$, then $\mathbf{\Sigma}_{t}$
stays as a diagonal matrix with identical diagonal entries, i.e., $\mathbf{\Sigma}_{t}=\tau_{t}^2\mathbf{I}$, for $t\geq1$, where $\tau_{t}$ is determined by
\begin{align}
\tau_{t+1}^2 = \sigma_w^2 + \frac{N}{L}\mathrm{MSE}(\tau_t).
\end{align}
If only $p_{G}(g_n)$ is known at the BS, $\mathrm{MSE}(\tau_t)$ is
given by
\begin{align}
\mathrm{MSE}(\tau_t)=&\int_{0}^{\infty}\frac{a\lambda g_n^{2}\tau_t^2Q(g_n)}{g_n^{\gamma}(g_n^2+\tau_t^2)}dg_n \nonumber \\
&+\int_{0}^{\infty}\frac{\mu_{M}(s)-\nu_{M}^{2}(s)\xi_{M}^{-1}(s)}{\Gamma(M+1)/(\lambda as^M)}ds,
\end{align}
where functions $\mu_{i}(s)$ , $\nu_i(s)$ and $\xi_{i}(s)$ are defined in
\eqref{A.func.mu}, \eqref{A.func.nu}, and \eqref{eq.xi}, respectively, and $\Gamma(\cdot)$ is the Gamma function. If the
exact large-scale fading $g_n$ is known at the BS, $\mathrm{MSE}(\tau_t)$
is given by
\begin{align}
\mathrm{MSE}(\tau_t)=&\int_{0}^{\infty}\frac{a\lambda g_n^2\tau_t^2Q(g_n)}{g_n^{\gamma}(g_n^2+\tau_t^2)}dg_n\nonumber\\
&+\int_{0}^{\infty}\frac{a\lambda Q(g_n)g_n^4}{g_n^{\gamma}(g_n^2+\tau_t^2)}\left(1-\frac{\varphi_{M}(g_n^2\tau_t^{-2})}{\Gamma(M+1)}\right)dg_n,
\end{align}
where function $\varphi_{i}(s)$ is defined in \eqref{A.func.varphi}.
\end{proposition}
\begin{proof}
See Appendix~\ref{A:diag}.
\end{proof}

Note that $\mathbf{\Sigma}_{0}$ is the noise covariance matrix after the first matched filtering, which is indeed a diagonal matrix with identical diagonal entries. Based on Proposition~\ref{prop.evo}, the MMSE denoiser in \eqref{eq.vamp.nolsf.mmse} can be further simplified as
\begin{align}\label{eq.vamp.nolsf.mmse.simp}
\mathbb{E}[\mathbf{R}|\mathbf{\tilde{R}}^{t}=\mathbf{\tilde{r}}_{n}^{t}]=\mathbf{\tilde{r}}^t_n\frac{\nu_{M}(\|\mathbf{\tilde{r}}^{t}_n\|_2^2)}{\xi_{M}(\|\mathbf{\tilde{r}}^{t}_n\|_2^2)},
\end{align}
where $\nu_i(s)$ and $\xi_i(s)$ are defined in \eqref{A.func.nu} and \eqref{eq.xi}, respectively,
and the MMSE denoiser in \eqref{eq.vamp.mmse} can be simplified as
\begin{multline}\label{eq.vamp.mmse.simp}
\mathbb{E}[\mathbf{R}|\mathbf{\tilde{R}}^{t}=\mathbf{\tilde{r}}_{n}^{t},G=g_n] \\
= \frac{g_n^2(g_n^2+\tau_t^2)^{-1}\mathbf{\tilde{r}}^{t}_{n}}{1+\frac{1-\lambda}{\lambda}(\frac{g_n^2+\tau_t^2}{\tau_t^2})^{M}\exp(-\Delta\|\mathbf{\tilde{r}}^{t}_n\|_2^2)},
\end{multline}
where $\Delta$ is defined in (\ref{eq:Delta}).
Note that if we let $M=1$, \eqref{eq.vamp.nolsf.mmse.simp} and \eqref{eq.vamp.mmse.simp} reduce to the denoisers for the single-antenna case in \eqref{eq.mmseEst} and \eqref{eq.mmseEst.lsf}. As mentioned before, we can also pre-compute and store the functions $\nu_M(\cdot)$ and $\xi_M(\cdot)$ in \eqref{eq.vamp.nolsf.mmse.simp} as table lookup.

After the AMP algorithm has converged, we use the likelihood ratio test to perform the user activity detection. Recall that $\mathbf{\tilde{R}}^{t}=\mathbf{R}+\mathbf{U}^{t}$ where $\mathbf{U}^{t}$ follows complex Gaussian distribution. For the case where the large-scale fading is unknown, based on $\eqref{eq.vamp.nolsf.distri}$ and $p_{\mathbf{\tilde{R}}^{t}}$ derived in \eqref{A.dist.tildr.nolsf} in Appendix~\ref{A.mmse.v}, the likelihood probabilities given that the user is inactive and active are, respectively
\begin{align}
p_{\mathbf{\tilde{R}}^{t}|\mathbf{R}}(\mathbf{\tilde{r}}_n^t|\mathbf{R}= \mathbf{0}) &= \frac{\exp\left(-\|\mathbf{\tilde{r}}^{t}_n\|_2^2\tau_t^{-2}\right)}{\pi^M \tau_t^{2M}},\label{eq.vamp.like1}\\
p_{\mathbf{\tilde{R}}^{t}|\mathbf{R}}(\mathbf{\tilde{r}}_n^t|\mathbf{R}\neq
\mathbf{0}) &=
\int_{0}^{\infty}\frac{ag_n^{-\gamma}Q(g_n)}{\pi^M(g_n^2+\tau_t^2)^M} \nonumber \\
 & \qquad \qquad \exp\left(\frac{-\|\mathbf{\tilde{r}}^{t}_n\|_2^2}{g_n^2+\tau_{t}^2}\right)dg_n.\label{eq.vamp.like2}
\end{align}
For the case where the large-scale fading coefficient is known, noting that $\mathbf{R}$ follows a Beroulli-Gaussian distribution, and $\mathbf{\tilde{R}}^{t}$ follows a mixed Gaussian distribution, then the likelihood probabilities can be computed as, respectively
\begin{align}
p_{\mathbf{\tilde{R}}^{t}|\mathbf{R},G}(\mathbf{\tilde{r}}^{t}_n|\mathbf{R}=\mathbf{0},G=g_n)&= \frac{\exp\left(-\|\mathbf{\tilde{r}}^{t}_n\|_2^2\tau_t^{-2}\right)}{\pi^M \tau_t^{2M}},\label{eq.vamp.like1}\\
p_{\mathbf{\tilde{R}}^{t}|\mathbf{R},G}(\mathbf{\tilde{r}}_n|\mathbf{R}\neq \mathbf{0},G=g_n) &= \frac{\exp\left(-\|\mathbf{\tilde{r}}^{t}_n\|_2^2(\tau_t^2+g_n^2)^{-1}\right)}{\pi^M (\tau_t^2+g_n^2)^{M}}.\label{eq.vamp.like2}
\end{align}
For both cases, we immediately obtain the $\operatorname{LLRs}$ as, respectively
\begin{align}
\label{Foad_LLR_vAMP}
&\operatorname{LLR} = \log\int_{0}^{\infty}\frac{a\tau_t^{2M} g_n^{-\gamma}Q(g_n)}{(g_n^2+\tau_t^2)^{M}}\exp\big(\|\mathbf{\tilde{r}}^{t}_n\|_2^2\Delta\big)dg_n,\\
&\operatorname{LLR}(g_n) = \log\Big(\frac{\tau_t^{2M}}{(\tau_t^2+g_n^2)^M}\exp\big(\|\mathbf{\tilde{r}}^{t}_n\|_2^2\Delta\big)\Big). \label{Foad_LLR_vAMP_lsf}
\end{align}
Observing that LLRs are monotonic in $\|\mathbf{\tilde{r}}^{t}_n\|_2$, we can set a threshold $l_n$ on $\|\mathbf{\tilde{r}}^{t}_n\|_2$ to perform the detection. When the large-scale fading is unknown at the BS, the probabilities of false alarm or missed detection are then given as, respectively
\begin{align}
P_F & = \int_{\|\mathbf{\tilde{r}}^{t}_n\|_2>l_n}\frac{\exp\left(-\|\mathbf{\tilde{r}}^{t}_n\|_2^2\tau_t^{-2}\right)}{\pi^M \tau_t^{2M}}d\mathbf{\tilde{r}}^{t}_n\nonumber\\
&\overset{(a)}{=}1-\frac{1}{\Gamma(M)}\bar{\gamma}\left(M,l_n^2\tau_t^{-2}\right),\label{eq.vamp.pf.lsf}
\end{align}
and
\begin{align}
P_M & =
\int_{\|\mathbf{\tilde{r}}^{t}_n\|_2<l_n}\int_{0}^{\infty}\frac{ag_n^{-\gamma}Q(g_n)}{\pi^M(g_n^2+\tau_t^2)^M}
\nonumber \\
& \qquad \qquad \qquad \qquad \qquad \exp\left(\frac{-\|\mathbf{\tilde{r}}^{t}_n\|_2^2}{g_n^2+\tau_{t}^2}\right)dg_n d\mathbf{\tilde{r}}^{t}_n\nonumber\\
&\overset{(b)}{=}\int_{0}^{\infty}\frac{ag_n^{-\gamma}Q(g_n)}{\Gamma(M)}\bar{\gamma}\left(M,l_n^2(g_n^2+\tau_t^{2})^{-1}\right)dg_n,
\label{eq.vamp.pm.lsf}
\end{align}
where $\bar{\gamma}(\cdot,\cdot)$ is the lower incomplete Gamma function, and $(a)$ and $(b)$ are simply obtained by noticing that the integral of $\mathbf{\tilde{r}}^{t}_n$ can be interpreted as the cumulative distribution function (cdf) of a $\chi^2$ distribution with $2M$ degrees of freedom since $\|\mathbf{\tilde{r}}^{t}_n\|_2^2$ can be regarded as a sum of the squares of $2M$ identical real Gaussian random variables. Using the same approach, when the large scale fading is known to the BS, the probabilities of false alarm and missed detection can be evaluated as
\begin{align}
P_F(g_n) &= \int_{\|\mathbf{\tilde{r}}^{t}_n\|_2>l_n} \frac{\exp\left(-\|\mathbf{\tilde{r}}^{t}_n\|_2^2\tau_t^{-2}\right)}{\pi^M \tau_t^{2M}}d\mathbf{\tilde{r}}^{t}_n,\nonumber\\
&=1-\frac{1}{\Gamma(M)}\bar{\gamma}\left(M,l_n^2\tau_t^{-2}\right),\label{eq.vamp.pf}\\
P_M(g_n) &= \int_{\|\mathbf{\tilde{r}}^{t}_n\|_2<l_n}\frac{\exp\left(-\|\mathbf{\tilde{r}}^{t}_n\|_2^2(\tau_t^2+g_n^2)^{-1}\right)}{\pi^M (\tau_t^2+g_n^2)^{M}}d\mathbf{\tilde{r}}^{t}_n\nonumber\\
&=\frac{1}{\Gamma(M)}\bar{\gamma}\left(M,l_n^2(g_n^2+\tau_t^{2})^{-1}\right).\label{eq.vamp.pm}
\end{align}
It is easy to verify that when $M=1$, \eqref{eq.vamp.pf.lsf}, \eqref{eq.vamp.pm.lsf}, \eqref{eq.vamp.pf} and \eqref{eq.vamp.pm} reduce to \eqref{eq.pf}, \eqref{eq.pm}, \eqref{eq.pf.lsf} and \eqref{eq.pm.lsf}, respectively.

Based on \eqref{eq.vamp.pf.lsf}, \eqref{eq.vamp.pm.lsf}, \eqref{eq.vamp.pf} and
\eqref{eq.vamp.pm}, we can design the threshold $l_n$ to achieve a trade-off
between the probability of false alarm and probability of missed detection.
The proposed thresholding
strategy in the single-antenna case can still be used in multiple-antenna
scenario.

\subsection{User Activity Detection by Parallel AMP-MMV}
The outline of the parallel AMP-MMV algorithm is as presented in Algorithm~\ref{Alg:1}. In this section, we adopt the parallel AMP-MMV algorithm for our problem setup; we present the expression of the denoiser, $\eta(\cdot,\mathbf{g},t,i,m)\triangleq [\eta_{t,i,m}(\cdot,g_1),\cdots,\eta_{t,i,m}(\cdot,g_N)]^{T}$ and the probability of a device being active based on the decision at the $m$th AMP-SMV stage, $\harpoonr{\pi}_{nm}$. Here we only discuss the case where the large-scale fading is known. The extension to the scenario where the large-scale fading is unknown is similar.

Since the parallel AMP-MMV algorithm employs $M$ parallel AMP-SMVs, the expression of the scalar denoiser for each AMP-SMV is in the form of the MMSE denoiser for the single-antenna case in \eqref{eq.mmseEst.lsf}. However, instead of using the prior $\lambda$ as the probability of being active for each user, the algorithm has access to a better estimate for the probability of activities as $\harpoonl{\pi}_{nm}$ as algorithm proceeds. Therefore, the expression for the MMSE denoiser can be written as
\begin{equation}
\eta_{t,i,m}(\tilde{x}_{nm}^{t,i},g_n) =
\frac{g_n^2(g_n^2+\tau_{t,i}^2)^{-1}\tilde{x}_{nm}^{t,i}}{1+\frac{1-\harpoonl{\pi}_{nm}}{\harpoonl{\pi}_{nm}}\frac{g_n^2+\tau_{t,i}^2}{\tau_{t,i}^2}\exp\left(\frac{-g_n^2|\tilde{x}_{nm}^{t,i}|^2}{\tau_{t,i}^2(g_n^2+\tau_{t,i}^2)}\right)},
\end{equation}
where $\tilde{x}_{nm}^{t,i}$ and ${x}_{nm}$ are the elements in the $n$th row and the $m$th column of $\tilde{\mathbf{X}}^{t,i}$ and ${\mathbf{X}}$, respectively. At the end of the $i$th outer iteration, i.e., $t=T$, the likelihood probabilities given that the user is inactive or active can be written as
\begin{align}
p(\tilde{x}_{nm}^{T,i}|X=0,G=g_n) &= \frac{\exp(-|\tilde{x}_{nm}^{T,i}|^2\tau_{T,i}^{-2})}{\pi\tau_{T,i}^2} \label{Foad-P0}, \\
p(\tilde{x}_{nm}^{T,i}|X\neq0,G=g_n) &= \frac{\exp(-|\tilde{x}_{nm}^{T,i}|^2(g_n^2+\tau_{T,i}^2)^{-1})}{\pi(\tau_{T,i}^2+g_n^2)}.\label{Foad-P1}
\end{align}

Further, using equations \eqref{Foad-P0} and \eqref{Foad-P1}, the probability that user $n$ is active based on the decision at the $m$th AMP-SMV can be calculated as
\begin{align} \nonumber
\harpoonr{\pi}_{nm} &= \frac{p(\tilde{x}_{nm}^{T,i}|X \neq 0,
G=g_n)}{p(\tilde{x}_{nm}^{T,i}|X\neq 0,G=g_n)+p(\tilde{x}_{nm}^{T,i}|X = 0,G=g_n)} \\
&= \left(1+\frac{\tau_{T,i}^2+g_n^2}{\tau_{T,i}^2} \exp\left(\frac{-g_n^2 |\tilde{x}_{nm}^{T,i}|^2}{\tau_{T,i}^2(g_n^2+\tau_{T,i}^2)}\right)\right)^{-1}.
\end{align}

After the parallel AMP-MMV is terminated, we use likelihood ratio test to perform the user activity detection.
It can be shown that the $\operatorname{LLR}$ for user $n$ can be calculated as
\begin{equation}
\operatorname{LLR}(g_n) = \log\left(\frac{\tau_{T,I}^{2M}}{(\tau_{T,I}^2+g_n^2)^M}\exp\left(\frac{-g_n^2\sum_m |{\tilde{x}}_{nm}^{T,I}|^2}{(\tau_{T,I}^2+g_n^2)\tau_{T,I}^2}
\right)\right).
\end{equation}
It can be seen that the LLR expression for AMP-MMV algorithm is in a similar form as in \eqref{Foad_LLR_vAMP}. Therefore, with the same discussion, we can show that the probabilities of false alarm and missed detection can be further simplified in a form similar to \eqref{eq.vamp.pf} and \eqref{eq.vamp.pm}, respectively.
To have complete performance prediction analysis, we also need to determine $\tau^2_{T,I}$ in the parallel AMP-MMV algorithm. However, due to the soft information exchange between the antennas, deriving an analytic state evolution for $\tau^2_{t,i}$ is very challenging. The numerical experiments in Section~\ref{sec.numerical} show that the performance of parallel AMP-MMV is very similar to AMP with vector denoiser. This observation suggests that the parameter $\tau^2_{T,I}$ for the AMP-MMV algorithm should be similar to the final value of $\tau^2_{t}$ in AMP with vector denoiser.

We briefly discuss the complexities of AMP with vector denoiser and AMP-MMV. For both algorithms the computational complexities mainly lie in the matched filtering and residual calculation, which depend on the problem size as $O(NLM)$, at each iteration. The advantage of AMP-MMV is that parallel computation is allowed due to the division of MMV problem into several SMV problems.

\section{Simulation Results}
\label{sec.numerical}
We evaluate the performance of the proposed method in a cell of radius $R=1000$m with potential $N=4000$ users among which $200$ are active, i.e., $\lambda=0.05$. The channel fading
parameters are $\alpha=15.3, \beta=37.6$ and $\sigma_{\mathrm{SF}}=8$,
and the background noise is $-169$dBm/Hz over 10MHz.

We first consider the single-antenna case with only statistical knowledge of large-scale fading. Fig.~\ref{fig.roc.nolsf} shows the tradeoff between the probabilities of missed detection and
the false alarm of AMP with MMSE denoiser when the
pilot sequence length is set as $L=800$ and the transmit power is set as $5$dBm, $15$dBm, and $25$dBm. We see that the predicted $P_M$ and $P_F$ match the analysis very well. We also plot a lower bound using $\tau_\infty=\sigma_{w}$.
The lower bounds are very close to the actual performance, indicating that
after convergence AMP is able to almost completely eliminate multiuser
interference; the remaining error is dominated by the background noise.

\begin{figure}
\centerline{\epsfig{figure=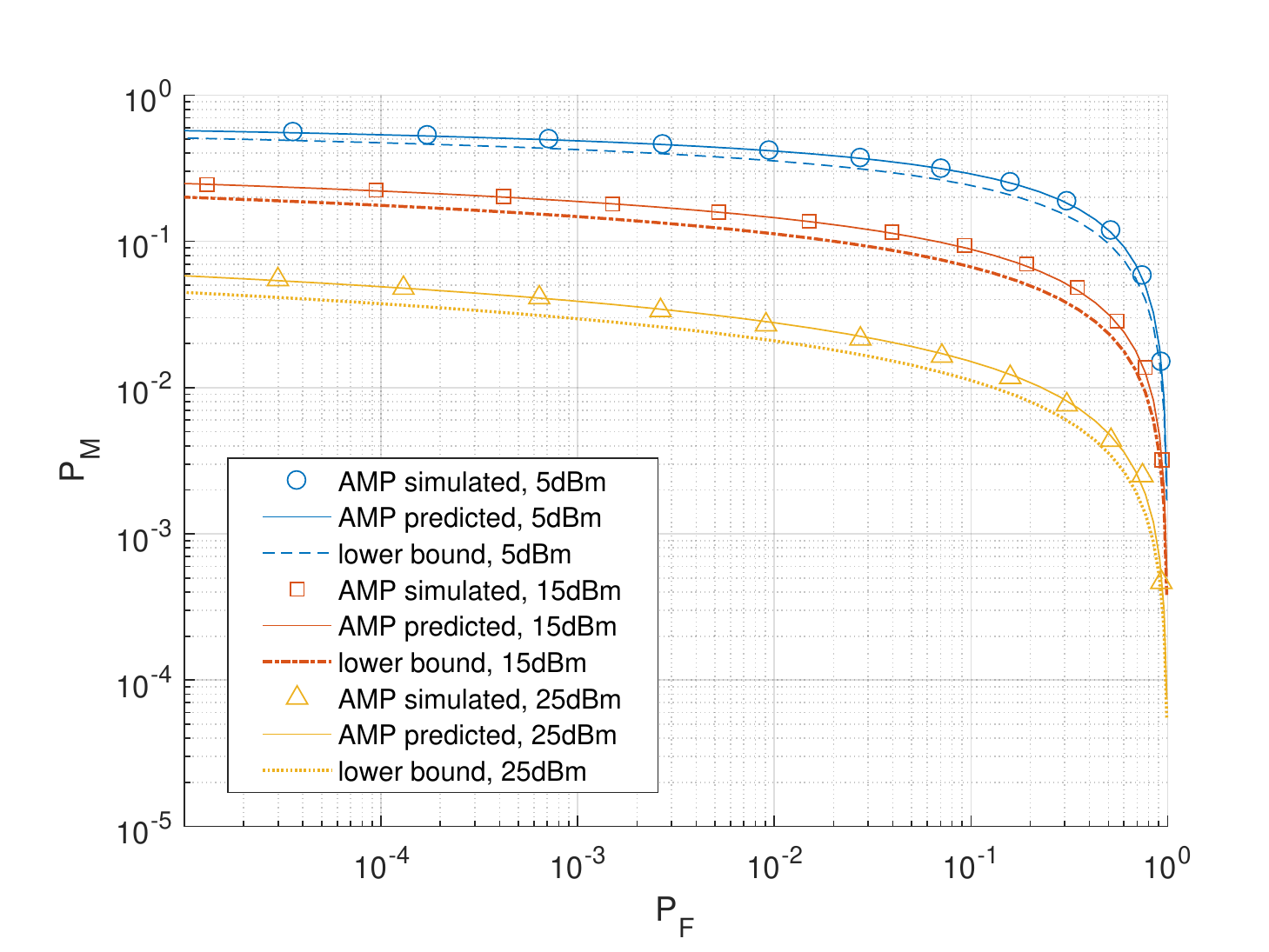,width=9.6cm}}
\caption{Performance of AMP based user activity detection with only statistical knowledge of the large-scale fading.}
\label{fig.roc.nolsf}
\end{figure}

Fig.~\ref{fig.roc} shows the performance of the AMP algorithm with MMSE denoiser when the exact large-scale fading coefficients are known. For comparison, the performance with only statistical knowledge of the large-scale fading is also demonstrated (only simulated performance is included since the predicted performance is almost the same as depicted in Fig.~\ref{fig.roc.nolsf}). Fig.~\ref{fig.roc} shows that the predicted curves match the simulated curves very well. More interestingly, it indicates that the performance improvement for knowing the exact large-scale fading coefficients is negligible
which suggests that knowing the distribution of the large-scale fading (rather than the exact value) is already enough for good user activity detection performance.

\begin{figure}
\centerline{\epsfig{figure=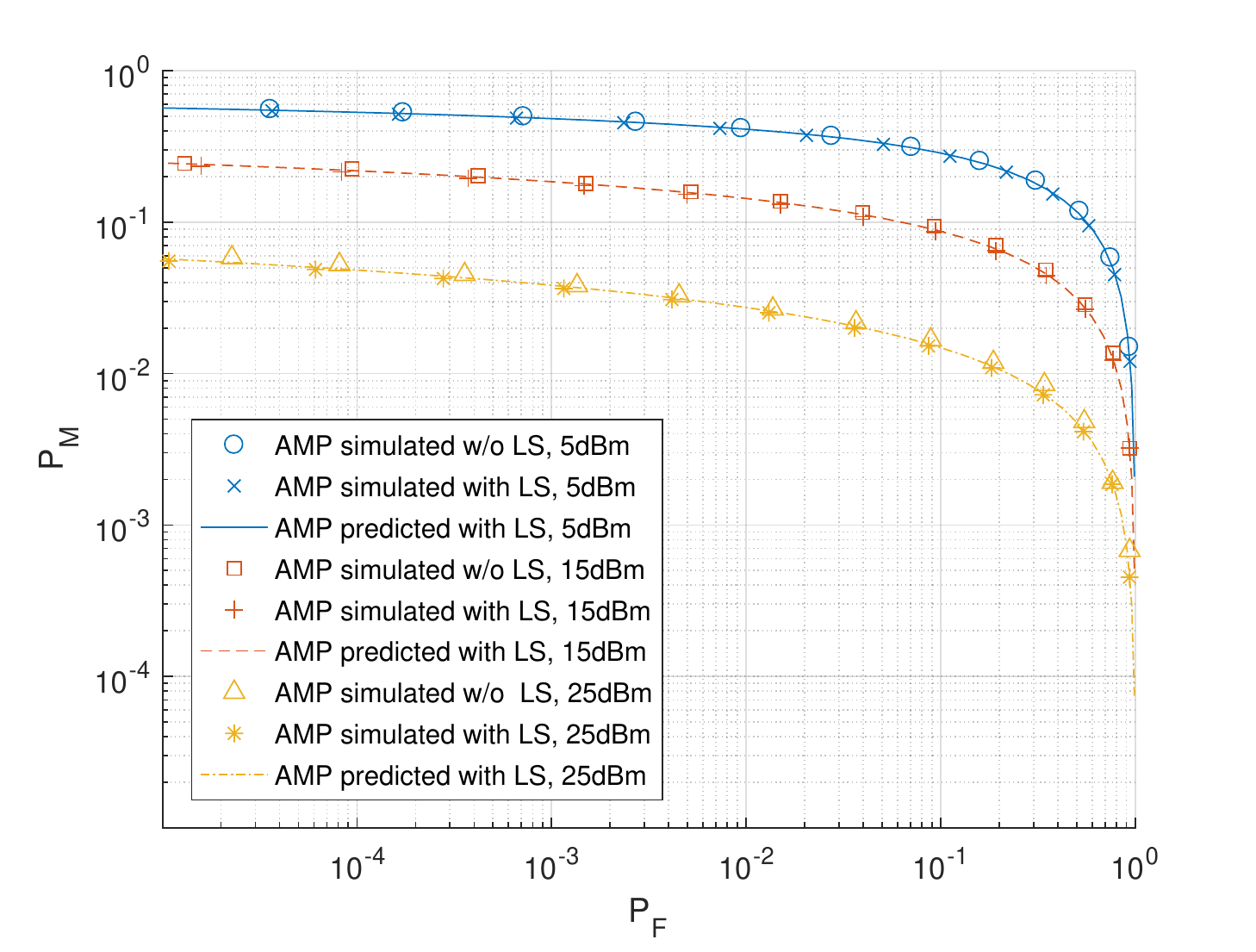,width=9.6cm}}
\caption{Performance of AMP based user activity detection with knowledge of the large-scale fading.}
\label{fig.roc}
\end{figure}

The next simulation compares the AMP algorithm with MMSE denoiser with two other algorithms widely used in compressed sensing: CoSaMP \cite{Needell2009}, and AMP but with soft thresholding denoiser \cite{Donoho2009}. Compare to AMP, CoSaMP is based on the matching pursuit technique. Compare to AMP with MMSE denoiser, AMP with soft thresholding does not exploit the statistical knowledge of $x_n$.
Fig.~\ref{fig.compare} shows that AMP with MMSE denoiser outperforms both CoSaMP and AMP with soft thresholding denoiser. This is partly due to the fact that both CoSaMP and AMP with soft thresholding denoiser do not exploit the the statistical knowledge of $x_n$. Note that AMP with soft thresholding implicitly solves the LASSO problem \cite{DonohoMalekiMontanari2011,Maleki2013}, i.e., the sparse signal recovery problem as an $\ell_1$-penalized least squares optimization. Therefore, the results in Fig.~\ref{fig.compare} indicate that AMP with MMSE denoiser also outperforms LASSO.

\begin{figure}
\centerline{\epsfig{figure=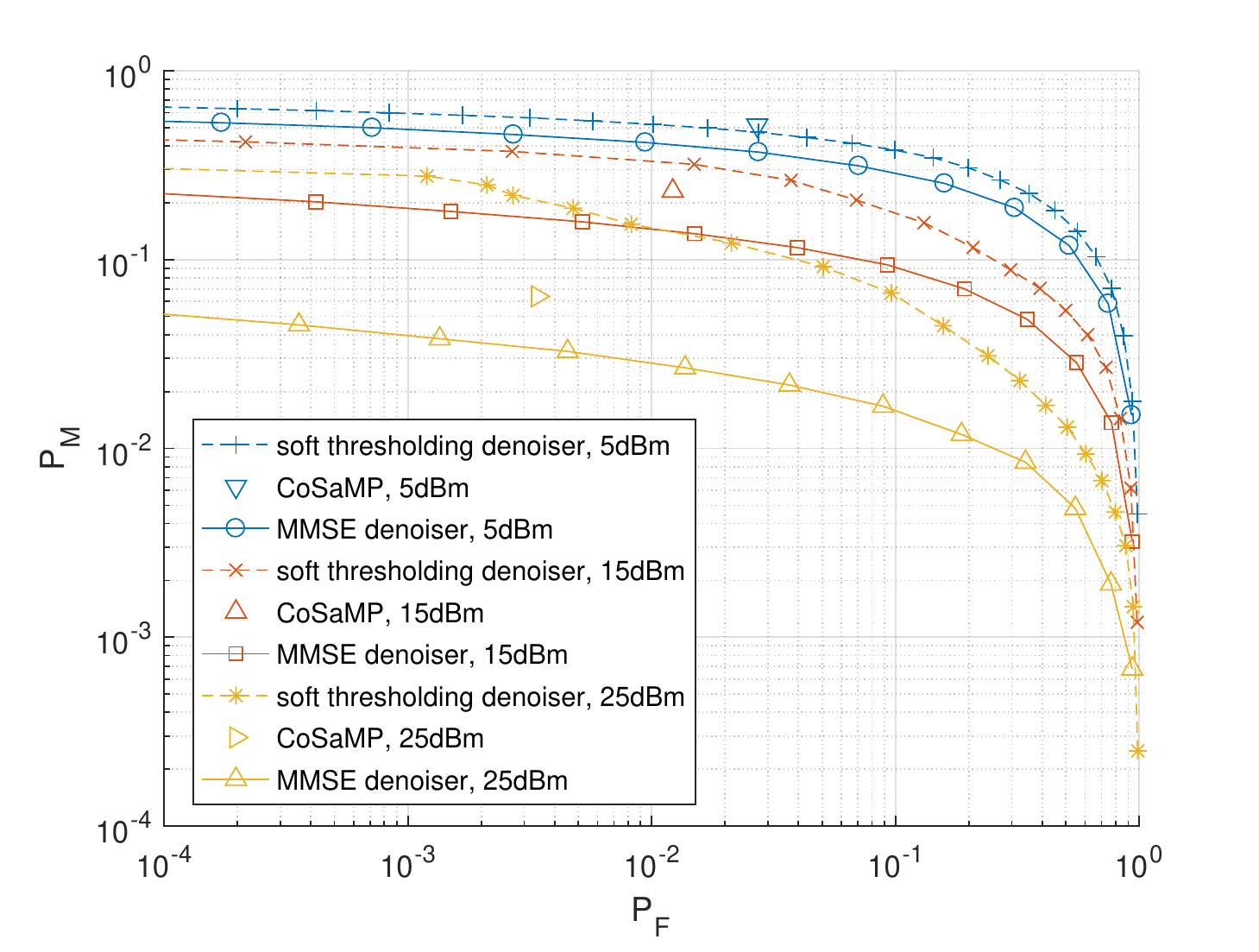,width=9.6cm}}
\caption{Performance comparison of AMP with MMSE denoiser, AMP with soft threshoding denoiser, and CoSaMP.}
\label{fig.compare}
\end{figure}

Fig.~\ref{fig.pfpm} compares the performance of the AMP algorithm with MMSE denoiser and the AMP algorithm with soft
thresholding denoiser as function of transmit power and pilot length.
For convenience, we set $P_F=P_M$ by properly choosing
the threshold $l$. We observe first that the MMSE denoiser outperforms
soft thresholding denoiser significantly,
but more importantly, we observe that the minimum $L$ needed to drive $P_F$ and $P_M$ to
zero as transmit power increases is between $300$ and $400$ for the MMSE denoiser, whereas the minimum $L$ is between $600$ and $800$ for the soft threshloding denoiser, indicating the clear
advantage of accounting for channel statistics in user activity
detector design.

\begin{figure}
\centerline{\epsfig{figure=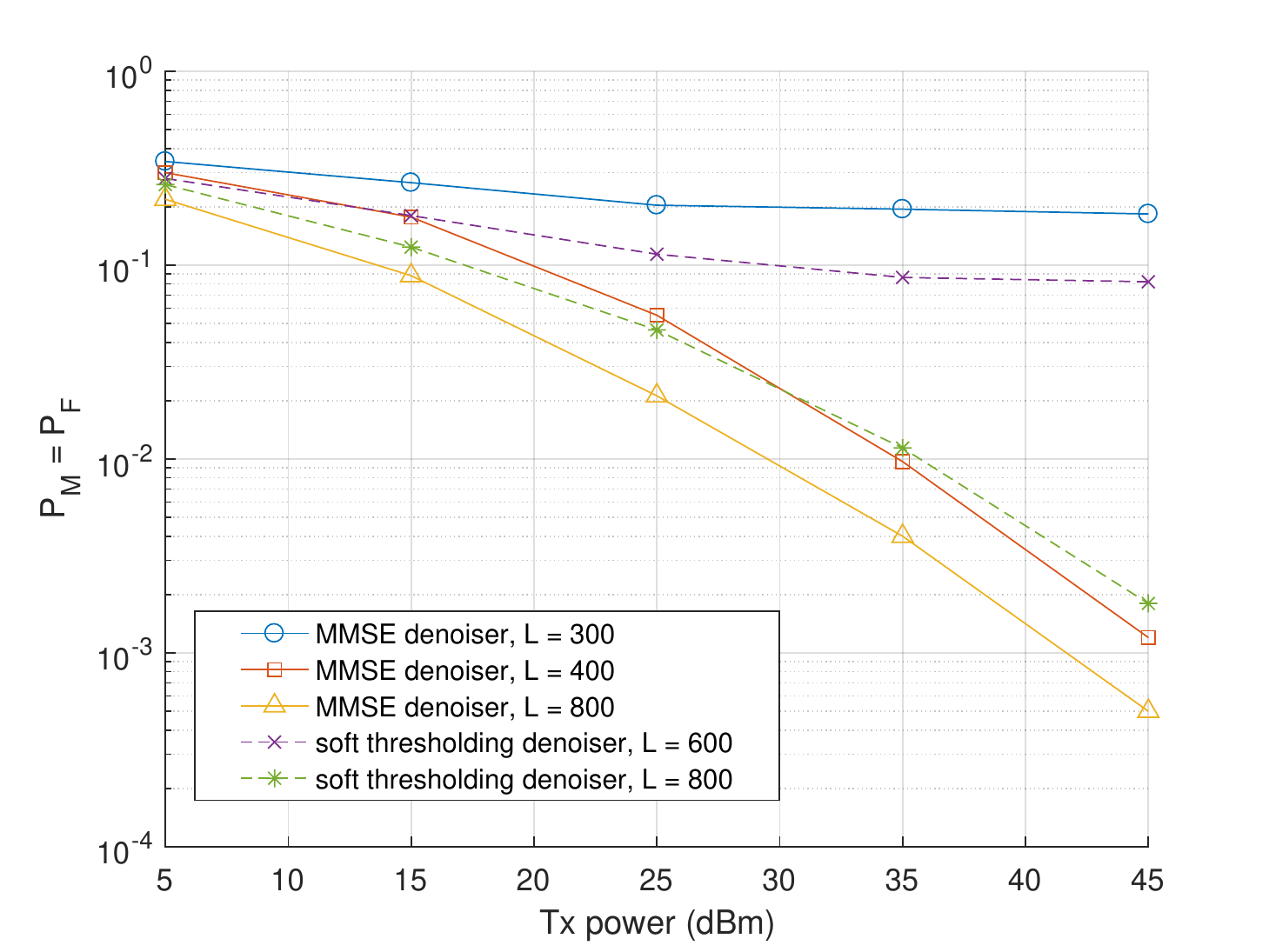,width=9.6cm}}
\caption{Impact of transmit power and length of pilot on user activity detection performance: MMSE denoiser vs. soft threholding denoiser.}
\label{fig.pfpm}
\end{figure}

Finally, we consider the multiple-antenna case assuming the knowledge of large-scale fading coefficients. Fig.~\ref{fig.mmv} illustrates the probabilities of false alarm and missed detection under different numbers of antennas for both AMP with vector denoiser and parallel AMP-MMV algorithms. For comparison, the single-antenna $M=1$ case is also included. Fig.~\ref{fig.mmv} shows that for AMP with vector denoiser, the simulated results match the predicted results very well. Further, it shows that the performances of AMP with vector denoiser and parallel
AMP-MMV
are approximately the same, indicating that although these two algorithms employ different strategies, they both exploit the statistical knowledge of the channel in the same way, resulting in similar performances.

\begin{figure}
\centerline{\epsfig{figure=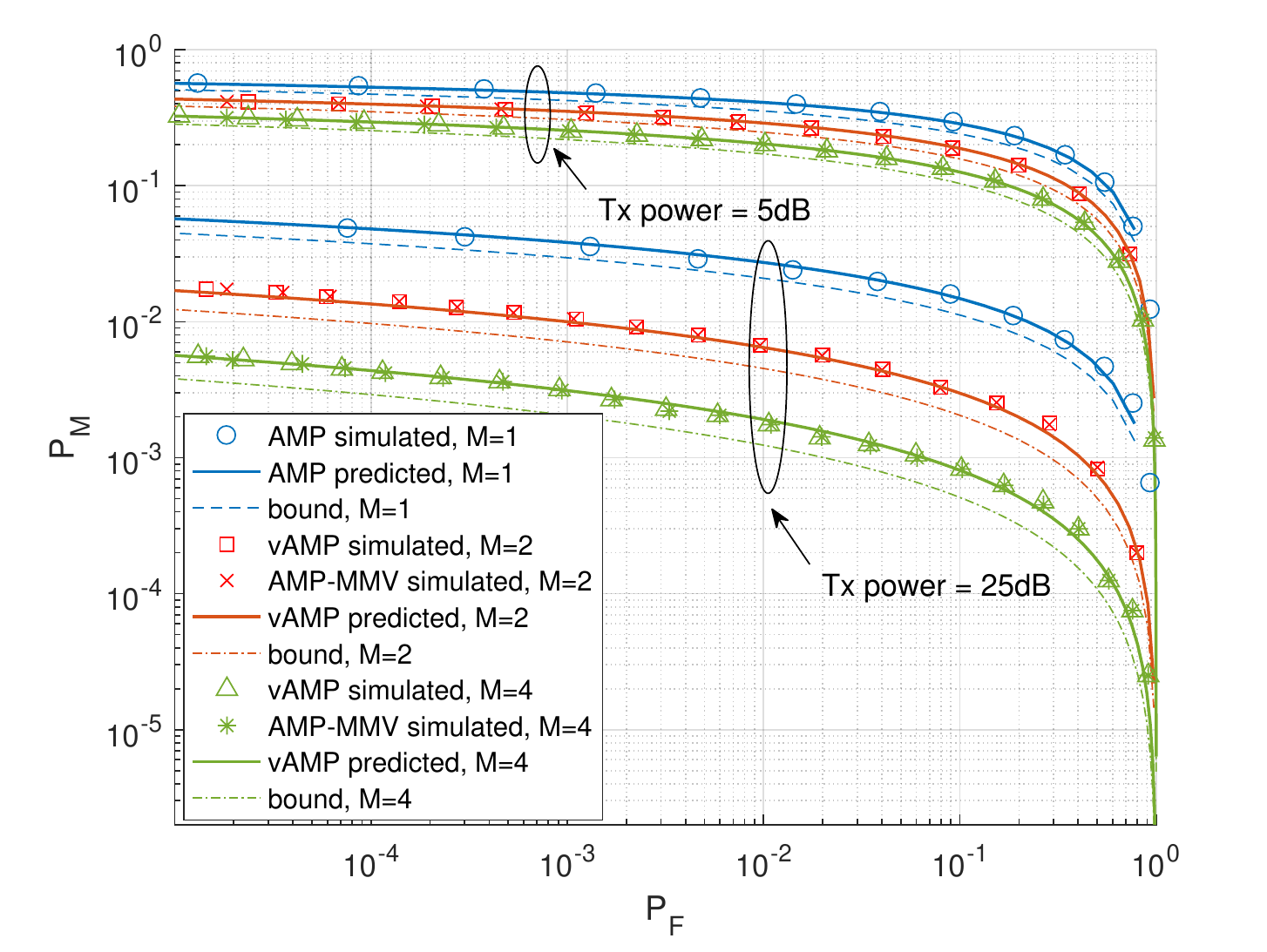,width=9.6cm}}
\caption{Performance of AMP with vector denoiser (vAMP) and AMP-MMV for user activity detection in the multiple-antenna case.}
\label{fig.mmv}
\end{figure}

The impact of the pilot length $L$ and the number of antennas $M$ on the probabilities of false alarm and missed detection as the transmit power increases is shown in Fig.~\ref{fig.pfpm.mmv}. We set $L=300, 600$, and $M=1, 2, 4$. We make $P_F=P_M$ for convenient comparison by properly choosing the threshold $l$.
Note that in the scenario where the exact large-scale fading coefficients are known, different users have the same probabilities of false alarm but different probabilities of missed detection. Thus we have to use the average probability of missed detection over all users. Fig.~\ref{fig.pfpm.mmv} shows that increasing $L$ or $M$ brings significant improvement. Specifically, when $L=300, M=1$, $P_F$ and $P_M$ tend to remain unchanged as the transmit power increases. However, by either increasing $L$ or increasing $M$, $P_F$ and $P_M$ can be driven to zero as the transmit power increases. In other words, the minimum $L$ required to drive $P_F$ and $P_M$ to zeros can be reduced by increasing $M$.

\begin{figure}
\centerline{\epsfig{figure=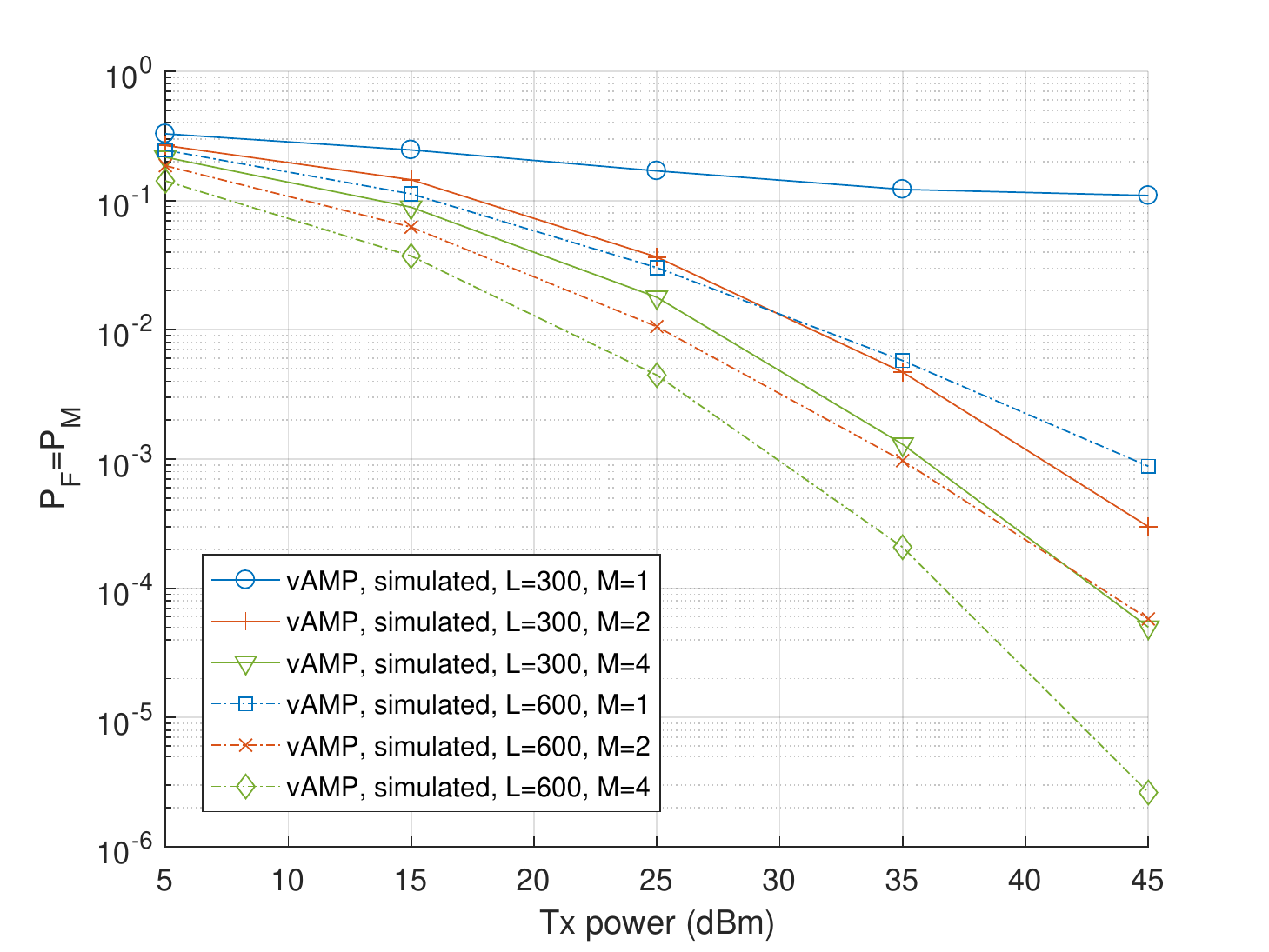,width=9.6cm}}
\caption{Impact of length of pilot, number of antennas, and transmit power.}
\label{fig.pfpm.mmv}
\end{figure}

\section{Conclusion}
\label{sec.con}
This work shows that compressed sensing is a viable strategy for sporadic device activity detection for massive connectivity applications with random non-orthogonal signature sequences. Specifically, we propose an AMP-based user activity detection algorithm by exploiting the statistics of the wireless channel for the uplink of a cellular system with a large number of potential users but only a small fraction of them
are active at any time slot.
We show that by using the state evolution, a performance characterization
in terms of the probabilities of false alarm and missed detection
can be accurately predicted.
In particular, we consider both
cases in which the BS is equipped with a single antenna or multiple antennas.
We present the designs of the MMSE denoisers in the scenarios where the large-scale fading is either available exactly or when only its statistics is available at the BS. For the multiple-antenna case, we adopt two AMP algorithms, AMP with vector denoiser and parallel AMP-MMV, to tackle the detection problem. We derive a performance analysis for both the single-antenna case and the multiple-antenna case. Simulation results validate the analysis, and show that exploiting the statistics of the channel in AMP denoiser design can significantly improve the detection threshold, and further deploying multiple antennas at the BS can also bring significant performance improvement.

\appendix
\subsection{Proof of Proposition \ref{P.Prop1}}
\label{A:probG}
Denote $d$, $x$, $y$, $z$ as the distance from a user to the BS, the shadowing in dB, the shadowing in linear scale, and the path-loss, respectively. Note that in appendices we slightly abuse some notations appeared in the paper due to the limited alphabet. However this should not cause confusion since they only used for derivations in the appendices. Denote $g\triangleq yz$ as the large-scale fading. We derive the pdfs of $d$, $x$, $y$, $z$ and $g$ as follows.

Assuming that all users are uniformly distributed in the cell with radius $R$, the pdf of $d$ is
\begin{align}
p_D(d)=\frac{2d}{R^2},\quad 0<d<R.
\end{align}
Since $x=-20\log_{10}(y)$ follows Gaussian distribution with zero mean and variance $\sigma_{\mathrm{SF}}^2$ , the pdf of $y$ can be derived as
\begin{align}
p_Y(y) =\frac{20}{\ln(10)\sqrt{2\pi}\sigma_{\mathrm{SF}}y}\exp\left(-\frac{200\ln^2(y)}{\ln^2(10)\sigma_{\mathrm{SF}}^2}\right).
\end{align}
By using $z=10^{-(\alpha+\beta\log_{10}d)/20}$, we get the pdf of $z$ as
\begin{align}
p_Z(z)=\frac{40}{R^2\beta}10^{-2\alpha/\beta}z^{-40/\beta-1},
\end{align}
where $z>10^{-(\alpha+\beta\log_{10}R)/20}$. The pdf of $g$ is
\begin{align}
p_G(g)&=\int |y|^{-1}p_Y(y)p_Z\left(y^{-1}g\right)dy\triangleq ag^{-\gamma}Q(g),
\end{align}
where $Q(g)\triangleq \int_{b\ln g+c}^{\infty}\exp(-s^2)ds$, $\gamma \triangleq 40/\beta+1$, and $a, b$ and $c$ are constants given in the proposition.

\subsection{Proof of Proposition \ref{P.Prop2}}
\label{A:probH}
Denote $g$ as the large-scale fading following \eqref{eq.distr.of.lsf}, and $x=x_R+jx_I$ as the Rayleigh fading, where $x_R$ and $x_I$ are the real and imaginary parts of $x$, respectively. Denote $h=h_R+jh_I$ as the channel coefficient accounting for both large-scale fading $g$ and Rayleigh fading $x$, i.e., $h=gx$. Then the pdf of $h$ is
\begin{align}
p_H(h)&=\int p_{H_R,H_I,G}(h_R,h_I,g)dg\nonumber\\
&=\int p_G(g)p_{X_R,X_I}\left(\frac{h_R}{g},\frac{h_I}{g}\right)\Big|\frac{\partial (x_R,x_I,g)}{\partial (h_R,h_I,g)}\Big| dg\nonumber\\\label{A.eq.h}
&=\frac{a}{\pi}\int_{0}^{\infty}g^{-\gamma-2}Q(g)\exp\left(\frac{-|h|^2}{g^2}\right)dg,
\end{align}
where $p_{H_R,H_I,G}(h_R,h_I,g)$ and $p_{X_R,X_I}(x_R,x_I)$ are the joint pdfs, and $|\frac{\partial (x_R,x_I,g)}{\partial (h_R,h_I,g)}|=g^{-2}$ is the Jacobian determinant. When
$g$ is known,
$p_{H|G}(h|g)$ follows the complex Gaussian distribution with zero mean and variance $g^2$.

\subsection{Proof of Proposition \ref{P.Prop3}}
\label{A:mmse}
We omit superscript $t$ and  subscript $n$ for notation simplicity. The conditional expectation of $X$ given $\tilde{X}=\tilde{x}$ can be expressed as
\begin{align}
\mathbb{E}[X&|\tilde{X}=\tilde{x}]\nonumber\\
=&\int x \frac{p_{X}(x)}{p_{\tilde{X}}(\tilde{x})}p_{\tilde{X}|X}(\tilde{x}|x)dx\nonumber\\
=&\iint_0^{\infty}f(g)x\exp\left(\frac{-|x|^2}{g^2}+\frac{-|\tilde{x}-x|^2}{\tau^2}\right)dgdx\nonumber\\
=&\iint_0^{\infty}f(g)x\exp\left(\frac{-|\tilde{x}|^2}{g^2+\tau^2}+\frac{-|x-\delta\tilde{x}|^2}{\delta \tau^2}\right)dgdx\nonumber\\
\overset{(a)}{=}&\frac{\lambda a\tilde{x}}{p_{\tilde{X}}(\tilde{x})\pi}\int_0^{\infty}\exp\left(\frac{-|\tilde{x}|^2}{g^2+\tau^2}\right)\frac{Q(g)g^{2-\gamma}}{(g^2+\tau^2)^2}dg,\label{A.eq.mmse}
\end{align}
where $f(g)\triangleq \lambda aQ(g)/(p_{\tilde{X}}(\tilde{x})\pi^2\tau^2 g^{\gamma+2})$, $\delta\triangleq g^2/(g^2+\tau^2)$, $(a)$ is obtained by using Gaussian integral of $x$. By substituting the expression of $p_{\tilde{X}}(\tilde{x})$ derived in Appendix \ref{A:probHN}, we get $\mathbb{E}[X|\tilde{X}=\tilde{x}]$ as in \eqref{eq.mmseEst}.

\subsection{Proof of Proposition \ref{P.Prop4}}
\label{A:probHN}
We omit superscript $t$ and subscript $n$ in the following. Note that $p_{\tilde{X}|X}(\tilde{x}|X=0)=p_{W}(\tilde{x})$ where random variable $W$ is defined as $W\triangleq\tau V$, and $p_{\tilde{X}|X}(\tilde{x}|X\neq 0)=p_{Y}(\tilde{x})$ where random variable $Y$ is defined as $Y\triangleq H+W$ with $H$ following the distribution in \eqref{eq.distr.of.channel}.
Since $V$ follows complex Gaussian distribution with zero mean and unit variance, we get
\begin{align}\label{A.eq.w}
p_{\tilde{X}|X}(\tilde{x}|X=0)=\frac{1}{\pi\tau^2}\exp\left(\frac{-|\tilde{x}|^2}{\tau^2}\right),
\end{align}
To compute $p_{\tilde{X}|X}(\tilde{x}|X\neq 0)$, we derive $p_{Y}(y)$ as follows
\begin{align}\label{A:eq.y}
p_{Y}(y)&\overset{(a)}{=}\int p_{H}(y-w)p_{W}(w)dw\nonumber\\
&\overset{(b)}{=}\int_{0}^{\infty}\frac{ag^{-\gamma}Q(g)}{\pi(g^2+\tau^2)}\exp\left(\frac{-|y|^2}{g^2+\tau^2}\right)dg,
\end{align}
where $(a)$ is obtained by $p_{Y,W}(y,w)=p_{H}(y-w)p_{W}(w)$, and $(b)$ is obtained by substituting \eqref{A.eq.h}.
Combine the results in \eqref{A.eq.w} and \eqref{A:eq.y} with $p_{\tilde{X}|X}(\tilde{x}|X\neq 0)=p_{Y}(\tilde{x})$, we get
\begin{multline}\label{A.eq.dist.xtild}
p_{\tilde{X}}(\tilde{x})=\frac{1-\lambda}{\pi\tau^2}\exp\left(\frac{-|\tilde{x}|^2}{\tau^2}\right)
\\
+\int_{0}^{\infty}\frac{\lambda aQ(g)\exp\left({-|\tilde{x}|^2/(g^2+\tau^2)}\right)}{\pi g^{\gamma}(g^2+\tau^2)}dg.
\end{multline}

\subsection{Proof of Proposition \ref{P.Prop5}}
\label{A:MSE}
We omit superscript $t$ and subscript $n$ for notation simplicity. The conditional variance of $X$ given $\tilde{X}=\tilde{x}$ is
\begin{align}\label{A.eq.var}
\mathrm{Var}(X|\tilde{X}=\tilde{x})&=\mathbb{E}\big[|X|^2\big|\tilde{X}=\tilde{x}\big]-\big|\mathbb{E}[X|\tilde{X}=\tilde{x}]\big|^2.
\end{align}
Since we have derived $\mathbb{E}[X|\tilde{X}=\tilde{x}]$ in \eqref{A.eq.mmse}, we only need to derive $\mathbb{E}\big[|X|^2\big|\tilde{X}=\tilde{x}\big]$, which can be expressed as
\begin{align}
\mathbb{E}\big[|X&|^2\big|\tilde{X}=\tilde{x}\big]\nonumber\\
= &\int\frac{|x|^2p_{X}(x)}{p_{\tilde{X}}(\tilde{x})}p_{\tilde{X}|X}(\tilde{x}|x)dx\nonumber\\
=&\iint_0^{\infty}
f(g)|x|^2\exp\left(\frac{-|x|^2}{g^2}+\frac{-|\tilde{x}-x|^2}{\tau^2}\right)dgdx\nonumber\\
=&\iint_0^{\infty}f(g)|x|^2\exp\left(\frac{-|\tilde{x}|^2}{g^2+\tau^2}+\frac{-|x-\delta\tilde{x}|^2}{\delta \tau^2}\right)dgdx\nonumber\\
\overset{(a)}{=}&\int_0^{\infty}\frac{\lambda aQ(g)\hat{f}(g)}{g^{\gamma}p_{\tilde{X}}(\tilde{x})\pi}\exp\left(\frac{-|\tilde{x}|^2}{g^2+\tau^2}\right)dg,
\end{align}
where $f(g)$ is defined in Appendix \ref{A:mmse}, $(a)$ is obtained by using Gaussian integral of $x$,
$\hat{f}(g)\triangleq\tau^2g^2/(g^2+\tau^2)^2+|\tilde{x}|^2g^4/(g^2+\tau^2)^3$. By using \eqref{A.eq.var}, \eqref{A.eq.dist.xtild}, \eqref{A.eq.mmse}, and $\mathrm{MSE}(\tau)=\int \mathrm{Var}(X|\tilde{X}=\tilde{x})p_{\tilde{X}}(\tilde{x})d\tilde{x}$ with some algebraic manipulations, we obtain \eqref{eq.mseofmmse}.

\subsection{Proof of Proposition \ref{P.Prop6}}
\label{A.MSE.lsf}
Similar to Appendix \ref{A:MSE}, we first derive the conditional expectation $\mathbb{E}\big[|X|^2|\tilde{X}=\tilde{x},G=g\big]$ as
\begin{align}
\mathbb{E}\big[|X|^2|\tilde{X}=\tilde{x},G=g\big]=\frac{\lambda\hat{f}(g)\exp(-|\tilde{x}|^2/(g^2+\tau^2))}{p_{\tilde{X}|G}(\tilde{x}|g)\pi},
\end{align}
where $\hat{f}(g)$ is defined in Appendix \ref{A:MSE}. Then based on \eqref{eq.mmseEst.lsf} and
$\mathrm{Var}(X|\tilde{X}=\tilde{x},G=g)=\mathbb{E}\big[|X|^2\big|\tilde{X}=\tilde{x},G=g\big]-\big|\mathbb{E}[X|\tilde{X}=\tilde{x},G=g]\big|^2$, we have
\begin{align}
\mathrm{MSE} &= \iint \mathrm{Var}(X|\tilde{X}=\tilde{x},G=g)p_{\tilde{X}|G}(\tilde{x}|g)p_{G}(g)d\tilde{x}dg\nonumber\\
&\overset{(a)}{=}\int_0^{\infty}\left(\frac{\lambda g^2\tau^2}{g^2+\tau^2}+\frac{\lambda g^4\big(1-\varphi_1(g^2\tau^{-2})\big)}{g^2+\tau^2}\right)p_{G}(g)dg,
\end{align}
where $(a)$ is obtained by using Gaussian integral over $\tilde{x}$. After plugging $p_{G}(g)$, we finally get \eqref{eq.mseofmmse.lsf}.

\subsection{Proof of Proposition \ref{P.Prop8}}
\label{A.mmse.v}
We omit superscript $t$ and subscript $n$. For the case where the large scale fading is unknown, the proof is similar to Appendix \ref{A:mmse} except that we need to deal with random vectors rather than random scalars. By using
\begin{align}\label{A:eq.mmse.v}
\mathbb{E}[\mathbf{R}|\mathbf{\tilde{R}}=\mathbf{\tilde{r}}]
&=\int \mathbf{r} \frac{p_{\mathbf{R}}(\mathbf{r})}{p_{\mathbf{\tilde{R}}}(\mathbf{\tilde{r}})}p_{\mathbf{\tilde{R}}|\mathbf{R}}(\mathbf{\tilde{r}}|\mathbf{r})d\mathbf{r},
\end{align}
where $p_{\mathbf{\tilde{R}}}(\mathbf{\tilde{r}})$ can be derived as
\begin{multline}\label{A.dist.tildr.nolsf}
p_{\mathbf{\tilde{R}}}(\mathbf{\tilde{r}})=\frac{1-\lambda}{(\pi\tau^2)^{M}}\exp\left(\frac{-\|\mathbf{\tilde{r}}\|_2^2}{\tau^2}\right) \\
+\int_{0}^{\infty}\frac{\lambda aQ(g)\exp\left({-\|\mathbf{\tilde{r}}\|_2^2/(g^2+\tau^2)}\right)}{\pi^{M} g^{\gamma}(g^2+\tau^2)^{M}}dg.
\end{multline}
By plugging $p_{\mathbf{\tilde{R}}}(\mathbf{\tilde{r}})$ and $p_{\mathbf{R}}(\mathbf{r})$ into \eqref{A:eq.mmse.v}, and using multivariate Gaussian integral of $\mathbf{r}$ with some algebraic  manipulations, we can obtain \eqref{eq.vamp.nolsf.mmse}. For the case where the large-scale fading is known, the conditional expectation can be found in \cite{Kim2011}.

\subsection{Proof of Proposition \ref{prop.evo}}
\label{A:diag}
Since the proofs for both cases are similar, in the following we focus on the case where the large-scale fading is known. We omit subscript $n$, and use $t$ as \emph{subscript} instead of superscript for convenience. We use induction by assuming $\mathbf{\Sigma}_{t}=\tau_t^2\mathbf{I}$ holds. To evaluate the right hand side of \eqref{eq.stateevo.multiple}, we first derive the distribution of $\mathbf{\tilde{R}}_{t}$ based on \eqref{eq.vamp.z} as
\begin{align}
p_{\mathbf{\tilde{R}}_{t}|G}(\mathbf{\tilde{r}}_t|g)=\frac{\lambda\exp\big(-\|\mathbf{\tilde{r}}_t\|_2^2(g^2+\tau_t^2)^{-1}\big)}{\pi^{M}(g^2+\tau_t^2)^M}\phi(\mathbf{\tilde{r}}_t),
\end{align}
where $\phi(\mathbf{\tilde{r}}_t)\triangleq 1+(1-\lambda)(1+g^2\tau_t^{-2})^{M}\exp(-\Delta\|\mathbf{\tilde{r}}_t\|_2^2)/\lambda$, and $\Delta$ is defined in \eqref{eq:Delta}. We then compute the conditional covariance matrix of $\mathbf{R}_t$ given $\mathbf{\tilde{R}}_t=\mathbf{\tilde{r}}_t$ and $G=g$ as
\begin{align}\label{A:cov_matrix}
\mathrm{Cov}=&\mathbb{E}[\mathbf{R}_{t}^T(\mathbf{R}_{t}^T)^{*}|\mathbf{\tilde{R}}_t=\mathbf{\tilde{r}}_t,G=g]\nonumber\\
&-\mathbb{E}[\mathbf{R}_{t}^T|\mathbf{\tilde{R}}_t=\mathbf{\tilde{r}}_t,G=g]\big(\mathbb{E}[\mathbf{R}_{t}^T|\mathbf{\tilde{R}}_t=\mathbf{\tilde{r}}_t,G=g]\big)^{*}\nonumber\\
=&\frac{g^2\tau_t^2\phi^{-1}(\mathbf{\tilde{r}}_t)}{g^2+\tau_t^2}\mathbf{I}+\frac{\phi^{-1}(\mathbf{\tilde{r}}_t)-\phi^{-2}(\mathbf{\tilde{r}}_t)}{g^{-4}(g^2+\tau_t^2)^2}\mathbf{\tilde{r}}_{t}^{T}(\mathbf{\tilde{r}}_{t}^{T})^{*}.
\end{align}
Then by taking the expectation over $\mathbf{\tilde{R}}_{t}$, we obtain
\begin{align}\label{A:cov_matrix_int}
\mathbb{E}_{\mathbf{\tilde{R}}_{t}|G}[\mathrm{Cov}] = \int  p_{\mathbf{\tilde{R}}_{t}|G}(\mathbf{\tilde{r}}_t|g)\mathrm{Cov}d\mathbf{\tilde{r}}_t,
\end{align}
which is a diagonal matrix due to fact that the off-diagonal element is an integral of an odd function over a symmetric interval, which is zero. Furthermore, it is easy to observe from the integral that the diagonal elements of $\mathbb{E}_{\mathbf{\tilde{R}}_{t}|G}[\mathrm{Cov}]$ are identical. Note that when MMSE denoiser is employed, the right hand side of \eqref{eq.stateevo.multiple} can be rewritten as $\mathbb{E}\big[\mathbf{D}_{t}\mathbf{D}_{t}^{*}\big]=\mathbb{E}_{G}\big[\mathbb{E}_{\mathbf{\tilde{R}}_{t}|G}[\mathrm{Cov}]\big]$, which leads to the result that $\mathbb{E}\big[\mathbf{D}_{t}\mathbf{D}_{t}^{*}\big]$ is also a diagonal matrix with identical diagonal elements.

In the following, we derive an explicit expression of $\mathbb{E}\big[\mathbf{D}_{t}\mathbf{D}_{t}^{*}\big]$. To this end, we first compute $\mathbb{E}_{\mathbf{\tilde{R}}_{t}|G}[\mathrm{Cov}]$. Denote $c_{i}$ as the $i$th diagonal entry of $\mathbb{E}_{\mathbf{\tilde{R}}_{t}|G}[\mathrm{Cov}]$, and $\tilde{r}_{t,i}$ is the $i$th entry of $\mathbf{\tilde{r}}_{t}$. Based on \eqref{A:cov_matrix} and \eqref{A:cov_matrix_int}, we have
\begin{align}
c_i =&\int \frac{g^2\tau_t^2}{g^2+\tau_t^2}\cdot\frac{\lambda\exp\big(-\|\mathbf{\tilde{r}}_t\|_2^2/(g^2+\tau_t^2)\big)}{\pi^{M}(g^2+\tau_t^2)^M}d\mathbf{\tilde{r}}_t\nonumber\\
& \qquad
+ \int \frac{|\tilde{r}_{t,i}|^2\big(1-\phi^{-1}(\mathbf{\tilde{r}}_t)\big)}{g^{-4}(g^2+\tau_t^2)^2}\cdot \nonumber \\
& \qquad \qquad
\frac{\lambda\exp\big(-\|\mathbf{\tilde{r}}_t\|_2^2/(g^2+\tau_t^2)\big)}{\pi^{M}(g^2+\tau_t^2)^M}d\mathbf{\tilde{r}}_t\nonumber\\
=&\frac{\lambda g^2\tau_t^2}{g^2+\tau_t^2} + \frac{\lambda g^4}{g^2+\tau_t^2}\left(1-\frac{\varphi_{M}(g^2\tau_t^{-2})}{\Gamma(M+1)}\right),
\end{align}
where the fist term of the last step is obtained by using Gaussian integral, the second term is obtained by integrating in spherical coordinates instead of Cartesian coordinates, and function $\varphi_i(s)$ is defined in \eqref{A.func.varphi}. As expected, $c_i$ does not depend on $i$, indicating that the diagonal elements are indeed identical. By replacing $g$ by $G$ and $c_i$ by $C$, we get $\mathbb{E}\big[\mathbf{D}_{t}\mathbf{D}_{t}^{*}\big]=\mathbb{E}_G[C]\mathbf{I}$, where $C$ is a random variable depends on $G$, and
\begin{multline}
\mathbb{E}_G[C] =\int_{0}^{\infty}\frac{aQ(g)}{g^{\gamma}}\cdot\frac{\lambda g^2\tau_t^2}{g^2+\tau_t^2}dg \\
+\int_{0}^{\infty}\frac{aQ(g)}{g^{\gamma}}\cdot\frac{\lambda g^4}{g^2+\tau_t^2}\left(1-\frac{\varphi_M(g^2\tau_t^{-2})}{\Gamma(M+1)}\right)dg.
\end{multline}
The state evolution in \eqref{eq.stateevo.multiple} is then simplified to
\begin{align}\label{eq.vamp.se}
\Sigma_{t+1}=\sigma_w^2\mathbf{I}+\frac{N}{L}\mathbb{E}_{G}[C\mathbf{I}]\triangleq \tau_{t+1}^2\mathbf{I},
\end{align}
which completes the induction.




\bibliography{chenbib}

\begin{IEEEbiography}[{\includegraphics[width=1in,height=1.25in,clip,keepaspectratio]{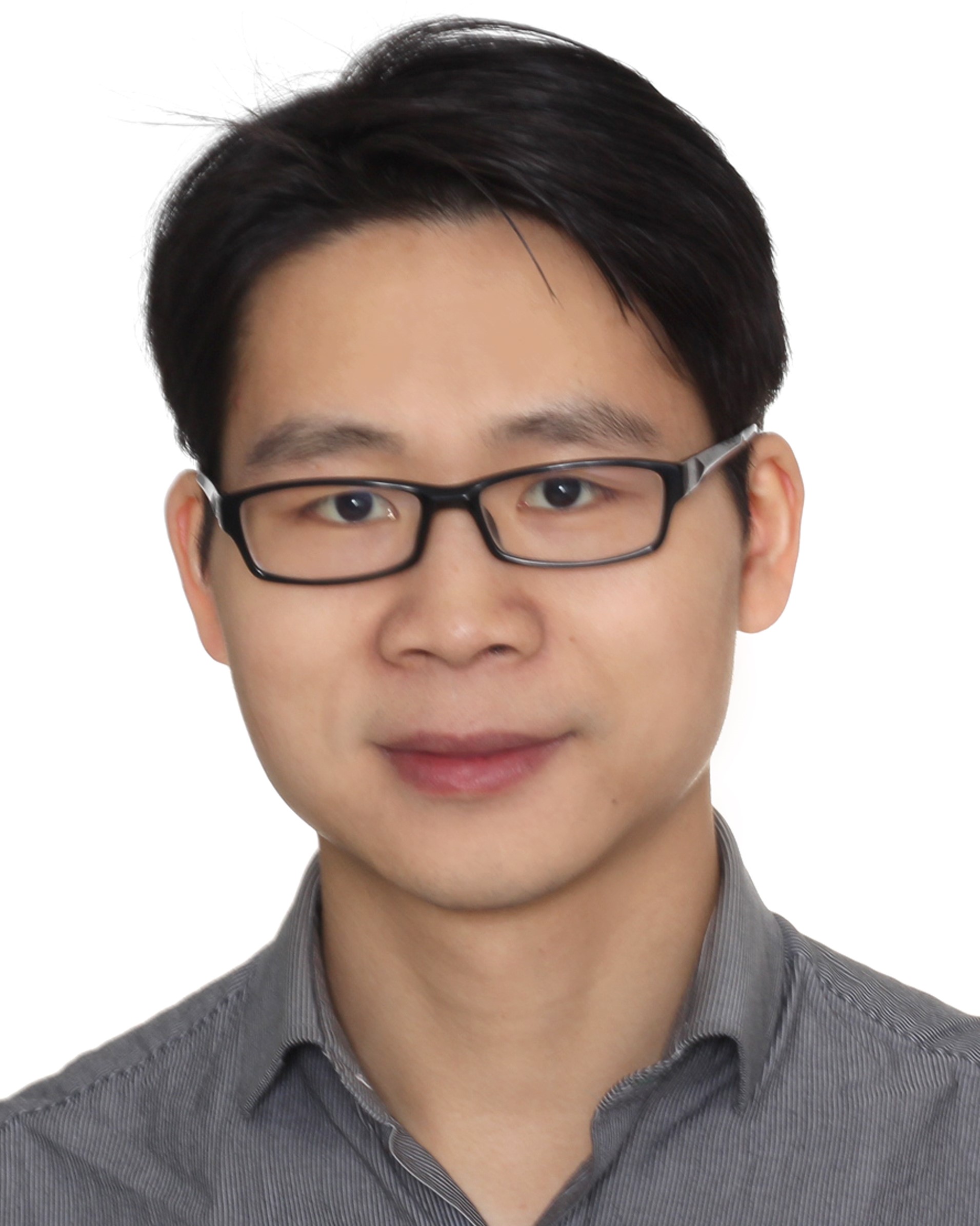}}]{Zhilin Chen}
(S'14) received the B.E. degree in electrical and information engineering and the M.E. degree in signal and information processing from Beihang University (BUAA), Beijing, China, in 2012 and 2015, respectively. He is currently pursuing the Ph.D. degree at the University of Toronto, Toronto, ON, Canada. His main research interests include wireless communication and signal processing.
\end{IEEEbiography}

\begin{IEEEbiography}[{\includegraphics[width=1in,height=1.25in,clip,keepaspectratio]{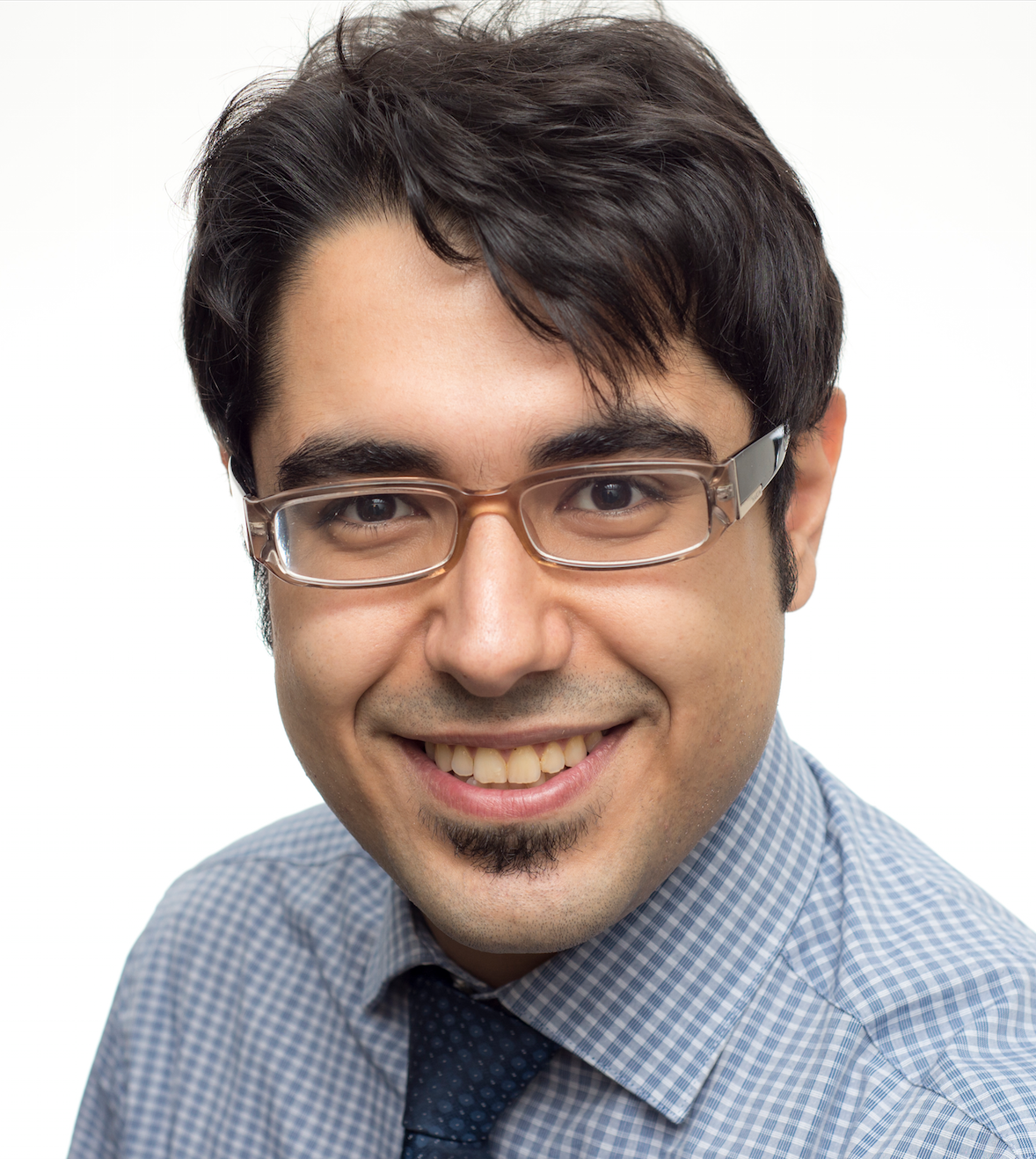}}]{Foad Sohrabi}
(S'13) received his B.A.Sc.\ degree in 2011 from the University of Tehran, Tehran, Iran, and his M.A.Sc.\ degree in 2013 from McMaster University, Hamilton, ON, Canada, both in Electrical and Computer Engineering. Since September 2013, he has been a Ph.D student at University of Toronto, Toronto, ON, Canada. Form July to December 2015, he was a research intern at Bell-Labs, Alcatel-Lucent, in Stuttgart, Germany. His main research interests include MIMO communications, optimization theory, wireless communications, and signal processing. He received an IEEE Signal Processing Society Best Paper Award in 2017.
\end{IEEEbiography}

\begin{IEEEbiography}[{\includegraphics[width=1in,height=1.25in,clip,keepaspectratio]{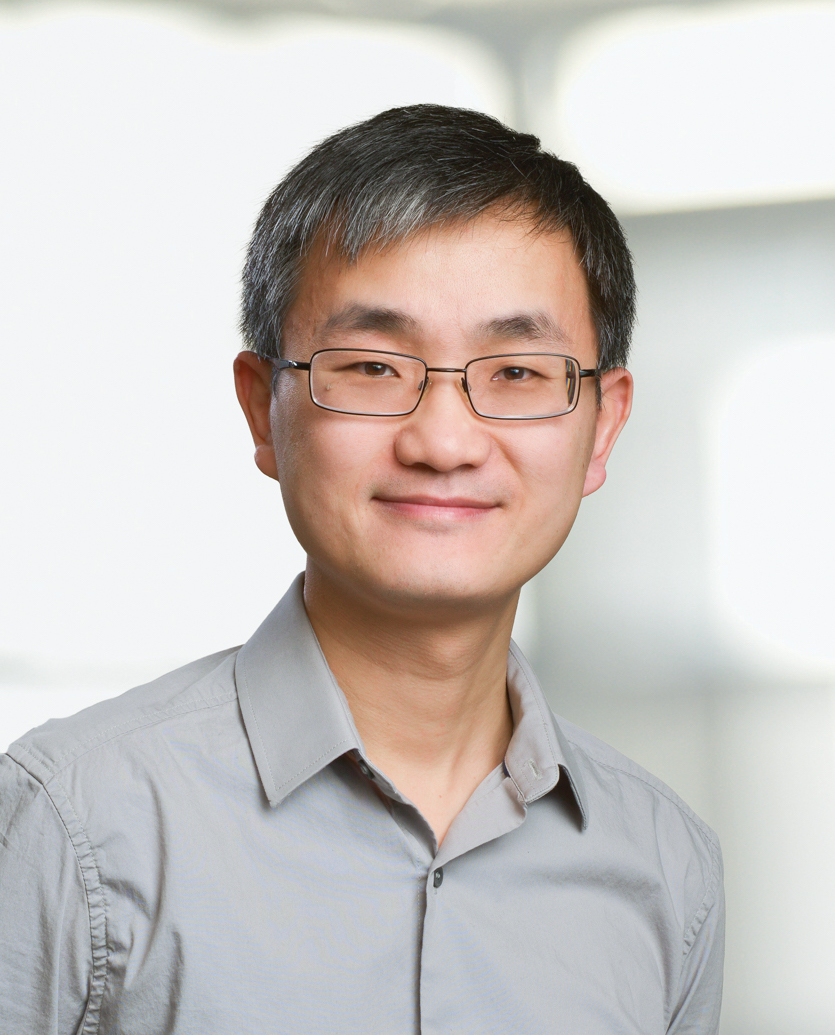}}]{Wei Yu}
(S'97-M'02-SM'08-F'14) received the B.A.Sc. degree in Computer Engineering and Mathematics from the University of Waterloo, Waterloo, Ontario, Canada in 1997 and M.S. and Ph.D. degrees in Electrical Engineering from Stanford University, Stanford, CA, in 1998 and 2002, respectively. Since 2002, he has been with the Electrical and Computer Engineering Department at the University of Toronto, Toronto, Ontario, Canada, where he is now Professor and holds a Canada Research Chair (Tier 1) in Information Theory and Wireless Communications. His main research interests include information theory, optimization, wireless communications and broadband access networks.

Prof. Wei Yu currently serves on the IEEE Information Theory Society Board of Governors (2015-20). He serves as the Chair of the Signal Processing for Communications and Networking Technical Committee of the IEEE Signal Processing Society (2017-18). He was an IEEE Communications Society Distinguished Lecturer (2015-16). He currently serves as an Area Editor of the IEEE \textsc{Transactions on Wireless Communications} (2017-19). He served as an Associate Editor for IEEE \textsc{Transactions on Information Theory} (2010-2013), as an Editor for IEEE \textsc{Transactions on Communications} (2009-2011), as an Editor for \textsc{IEEE Transactions on Wireless Communications} (2004-2007), and as a Guest Editor for a number of special issues for the \textsc{IEEE Journal on Selected Areas in Communications} and the \textsc{EURASIP Journal on Applied Signal Processing}. He was a Technical Program co-chair of the IEEE Communication Theory Workshop in 2014, and a Technical Program Committee co-chair of the Communication Theory Symposium at the IEEE International Conference on Communications (ICC) in 2012. Prof. Wei Yu received the IEEE Signal Processing Society Best Paper Award in 2017 and in 2008, a \textsc{Journal of Communications and Networks} Best Paper Award in 2017, a Steacie Memorial Fellowship in 2015, an IEEE Communications Society Best Tutorial Paper Award in 2015, an IEEE ICC Best Paper Award in 2013, the McCharles Prize for Early Career Research Distinction in 2008, the Early Career Teaching Award from the Faculty of Applied Science and Engineering, University of Toronto in 2007, and an Early Researcher Award from Ontario in 2006. Prof. Wei Yu is recognized as a Highly Cited Researcher. He is a Fellow of the Canadian Academy of Engineering.
\end{IEEEbiography}
\end{document}